\documentclass[runningheads,envcountsame]{llncs}

\usepackage{graphicx} 
\usepackage{amsmath, amssymb}
\usepackage{todonotes}
\usepackage{tikz}
\usepackage{comment}
\usetikzlibrary{automata,positioning,arrows.meta}

\definecolor{stateyellow}{RGB}{219,169,26}
\definecolor{stateteal}{RGB}{34,190,200}
\definecolor{edgepurple}{RGB}{120,0,200}
\definecolor{stateviolet}{RGB}{170,90,210}
\definecolor{edgeturquoise}{RGB}{70,210,200}

\DeclareMathOperator{\ta}{\mathtt{a}}
\DeclareMathOperator{\tb}{\mathtt{b}}

\usetikzlibrary{automata, positioning, arrows.meta}

\DeclareMathOperator{\N}{\mathbb{N}}

\title{Bandwidth of Nondeterministic Finite Automata}
\author{Da-Jung Cho\inst{1} \and Szil\'ard Zsolt Fazekas\inst{2} \and Daihei Ise\inst{3} \and Shinnosuke Seki\inst{3} \and \\ Wataru Tamehira\inst{3} \and Max Wiedenh\"oft\inst{4}}

\authorrunning{D-J.~Cho et al. }

\titlerunning{Bandwidth of Nondeterministic Finite Automata}

\institute{Department of Software and Computer Engineering, Ajou University, \\ Republic of Korea
\email{dajungcho@ajou.ac.kr} \and
Graduate School of Engineering Science, Akita University, Japan
\email{szilard.fazekas@ie.akita-u.ac.jp} \and
University of Electro-Communications, Tokyo, Japan \\
\email{\{daihei.ise, s.seki\}@uec.ac.jp}, \email{t2210411@gl.cc.uec.ac.jp} \and
Department of Computer Science, Kiel University, Kiel, Germany\\
\email{maw@informatik.uni-kiel.de}}

\begin{document}

\maketitle
\begin{abstract}
Co-transcriptional splicing generates RNA sequences from a DNA template by deleting subsequences nondeterministically. 
Recent work showed how to encode an NFA into such a template, but the construction requires deleting subsequences whose length grows with the distance between states, which makes such deletions unlikely under the local nature of co-transcriptional splicing.
We introduce $k$-bandwidth NFAs, in which transitions span at most $k$ states. 
These automata form a strict hierarchy of language classes. 
For finite languages, bandwidth~$2$ suffices, and bandwidth $1$ can be decided in polynomial-time when the language is presented as a list of words. 
Minimizing the bandwidth is NP-hard even for fixed $k \geq 2$.
\end{abstract}
\everypar{\looseness=-1}
    \section{Introduction}

Co-transcriptional splicing (Fig.~\ref{fig:cotranscriptional_splicing}) is a process in which an RNA sequence is edited by letting its subsequences fold into a hairpin and cut from its root while being sequentially synthesized from its DNA template, or \textit{transcribed}~\cite{MerkhoferHJ14}. 
Co-transcriptional \textit{folding}, by which hairpins to be spliced out form therein, has been proven programmable \textit{in vitro} for self-assembly of a 2-dimensional tile-like structure \cite{GearyRA14} and \textit{in papyro} for universal computation \cite{GearyMSS18}. 
The former architecture (\textit{RNA origami}) demonstrated how to program a chain of folding of constitutional hairpins and their interactions into an RNA sequence in such a way that the sequence co-transcriptionally folds according to the event chain. 
With rapid progress in experiments in mind, it is not unrealistic to envision applications of co-transcriptional splicing. 

\begin{figure}[tb]
\centering
\includegraphics[scale=0.35]{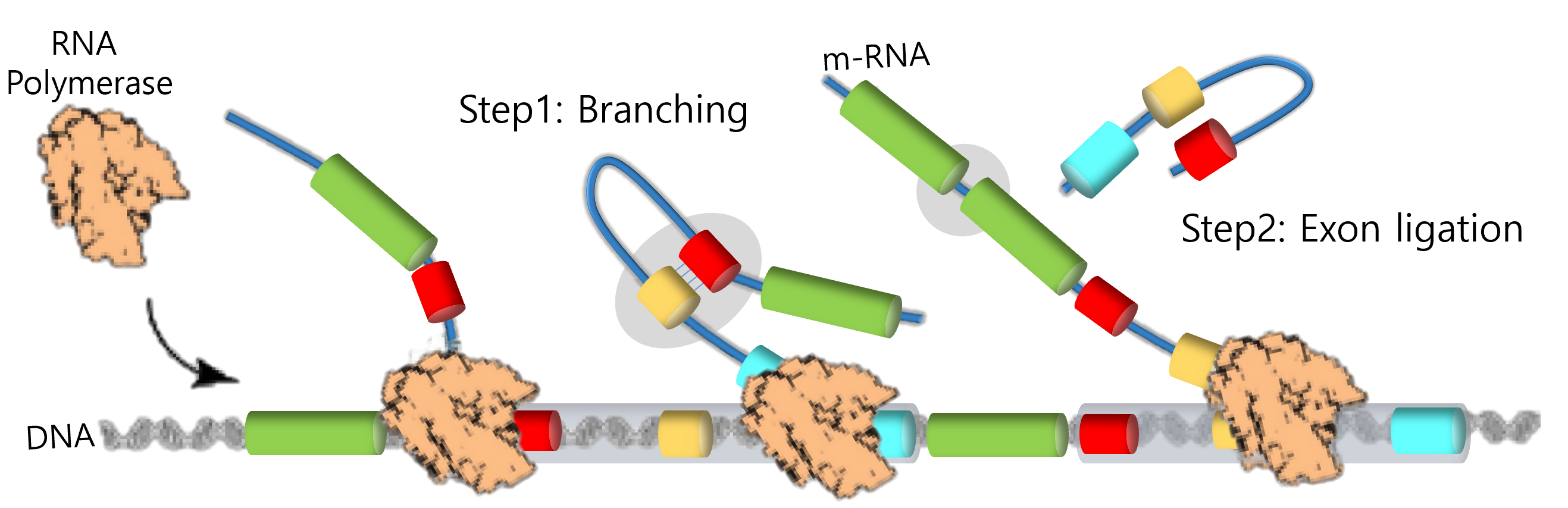}
\caption{Co-transcriptional splicing. 
While an RNA transcript is synthesized (\textit{transcribed}), a hairpin forms via hybridization of its complementary subsequences (colored in red and yellow) and spliced out. 
}
\label{fig:cotranscriptional_splicing}
\end{figure}

For instance, artificial molecular systems generate RNA sequences from DNA templates during transcription to reduce cost, since DNA is cheaper to synthesize.
As of now, many templates have to be ordered as target RNA sequences, but what if multiple targets could be accommodated on a single template? 

Co-transcriptional splicing can be modeled as a greedy deletion of hairpins one after another from a sequence~\cite{ChoFSW25_NC}. 
A hairpin is represented, using two sequences $x, \ell$ and an antimorphic involution $\theta$, as $x \ell \theta(x)$, where $x$ hybridizes with its ``Watson-Crick complement'' $\theta(x)$ into a stem, leaving $\ell$ as a loop. 
Using this model, we have recently demonstrated how to encode a nondeterministic finite automaton into a circular template $w$ in such a way that all words accepted by the NFA can be co-transcriptional spliced out of $w^*$ \cite{ChoFSW25}. 
The states $q_1, q_2, \cdots, q_n$ of a target NFA are encoded as sequences over a 4-letter alphabet and concatenated together with a special factor $s_e$ for termination into the template as $w = s_1 s_2 \cdots s_n s_e$ in such a way that 
\begin{enumerate}
    \item $s_i \cdots s_{j-1}$ can be decomposed as $h_1 a h_2$ for a letter $a$ and two hairpins $h_1, h_2$ as long as the NFA can transition from state $q_i$ to state $q_j$ by reading $a$; 
    \item $h_1$ is a prefix of $s_i$; 
    \item $h_2$ does not share its start point with any other hairpin. 
\end{enumerate}
Starting at the beginning of $s_1$, one of the outgoing transitions from $q_1$, or particularly, their prefix $h_1$'s is chosen nondeterministically and deleted, guiding the input head to the letter to be read, and then $h_2$ is deleted deterministically, guiding the head to the beginning of the encoding of destination state; the transition from $q_i$ to $q_j$ by $a$ is thus implemented. 
A drawback of this architecture is that $h_2$ can be arbitrarily long. 
The most serious threat is self-loop. 
In addition, the more outgoing transitions a state $q_i$ has, the longer the corresponding encoding~$s_i$ becomes. 
From a biological perspective, such long-range interactions are unlikely to occur. 
Co-transcriptional splicing operates in a local and greedy manner and acts on splice sites as soon as they are transcribed. 
This limits interactions between distant regions of the transcript and motivates restricting transitions to nearby states along the template.
This leads to a linearization problem: given an NFA, can its states be arranged along the template so that every transition spans only a bounded number of segments?

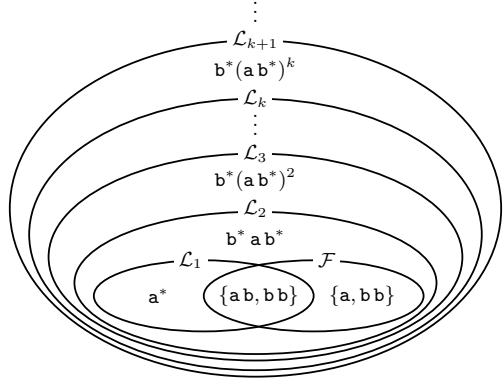
\begin{figure}[tb]
    \centering
    \resizebox{0.55\textwidth}{!}{
        \begin{tikzpicture}[outline/.style={draw, thick}]
        \node[fill=white, inner sep=0pt] at (0, 4.5) {$\vdots$};
  
        \draw[outline] (0, 1.35) ellipse (3.8cm and 2.6cm) node[fill=white, inner sep=2pt] at (0, 3.95) {$\mathcal{L}_{k+1}$};
        \node[fill=white, inner sep=0pt] at (0, 3.5) {$\tb^*(\ta\tb^*)^k$};

        \draw[outline] (0, 0.95) ellipse (3.5cm and 2.1cm) node[fill=white, inner sep=2pt] at (0, 3.05) {$\mathcal{L}_k$};
        \node[fill=none, inner sep=0pt] at (0, 2.75) {$\vdots$};

        \draw[outline] (0, 0.6) ellipse (3.2cm and 1.6cm) node[fill=white, inner sep=2pt] at (0, 2.2) {$\mathcal{L}_{3}$};
        \node[fill=white, inner sep=0pt] at (0, 1.8) {$\tb^*(\ta\tb^*)^2$};

        \draw[outline] (0, 0.2) ellipse (2.8cm and 1.1cm) node[fill=white, inner sep=2pt] at (0, 1.35) {$\mathcal{L}_{2}$};
        \node[fill=white, inner sep=0pt] at (0, 0.9) {$\tb^*\ta\tb^*$};

        \draw[outline] (-0.8, 0) ellipse (1.7cm and 0.55cm) node[fill=white, inner sep=2pt] at (-1.0, 0.55) {$\mathcal{L}_1$};
        \draw[outline] (0.9, 0) ellipse (1.7cm and 0.55cm) node[fill=white, inner sep=2pt] at (1.1, 0.55) {$\mathcal{F}$};
        \node[fill=white, inner sep=1.5pt] at (-1.5, 0) {$\ta^*$};
        \node[fill=white, inner sep=1.5pt] at (1.65, 0) {$\{\ta, \tb\tb\}$};
        \node[fill=white, inner sep=1.5pt] at (0.05, 0) {$\{\ta\tb, \tb\tb\}$};
    \end{tikzpicture}
    }
    \caption{A relationship among  $\mathcal{L}_1,\mathcal{L}_2,\ldots$ and $\mathcal{F}$, established as Proposition \ref{proposition:proper-hierarchy}, Lemma \ref{lemma:finite-existence-2BW}, and Lemma \ref{lemma:finite-language-not-in-L1}.}
    \label{fig:relationship}
\end{figure}

The goal of this paper is to introduce the notion of \emph{$k$-bandwidth NFAs}, whose states can be ordered in such a way that all transitions skip at most $k-1$ states in-between, and investigate their properties. 
We first show an infinite hierarchy among the class of such NFAs with respect to the bandwidth $k$ (Fig.~\ref{fig:relationship}). 
The class $\mathcal{F}$ of finite languages is of practical interest as no molecular system involves infinitely many kinds of RNA sequences. 
In this line, we will show how to convert an arbitrary finite language into a bandwidth-$2$ NFA that accepts it (Lemma~\ref{lemma:finite-existence-2BW}). 
This result turns out to be optimal in the sense that there exists a finite language that cannot be accepted by any bandwidth-1 NFA (Lemma~\ref{lemma:finite-language-not-in-L1}), but we provide a polynomial-time bandwidth-1 test for finite languages (Theorem~\ref{theorem:existence-bw-nfa-k1}).
Minimizing the bandwidth turns out to be NP-hard even for fixed $k \geq 2$ (Theorem~\ref{theorem:minization-hard-fixed-k}).

    \section{Preliminaries}

Let $\N$ denote the set of positive integers.
For some $m\in\N$, denote by $[m]$ the set $\{1,2,...,m\}$.
With $\Sigma$, we denote a finite set of symbols, called an \emph{alphabet}.
The elements of $\Sigma$ are called \emph{letters}.
A \emph{word} over $\Sigma$ is a finite sequence of letters from $\Sigma$.
With $\varepsilon$, we denote the \emph{empty word}.
The $i$-th element of a word $w$ is denoted by $w[i]$.
The set of all words over $\Sigma$ is denoted by $\Sigma^*$.
Additionally, set $\Sigma^+ = \Sigma^*\setminus\{\varepsilon\}$.
The \emph{length} of a word $w\in\Sigma^*$, i.e., the number letters in $w$, is denoted by $|w|$; hence, $|\varepsilon| = 0$.
For some $w\in\Sigma^*$, if we have $w = xyz$ for some $x,y,z\in\Sigma^*$, then we say that $x$ is a \emph{prefix}, $y$ is a \emph{factor}, and $z$ is a \emph{suffix} of $w$.
A nondeterministic finite automaton~(NFA) is a tuple $A = (\Sigma, Q, q_0,\delta, F)$ where $\Sigma$ is the input alphabet, $Q$ is the finite set of states, $\delta \colon Q \times \Sigma \rightarrow 2^Q$ is the multivalued transition function, $q_0 \in Q$ is the initial state, and $F \subseteq Q$ is the set of final states.
In the usual way, $\delta$ is extended as a function $Q \times \Sigma^* \rightarrow 2^Q$ and the language accepted by $A$ is $L(A) = \{ w \in \Sigma^* \mid \delta(q_0, w) \cap F \neq \emptyset \}$.

\begin{definition}[$k$-bandwidth NFA]\label{definition:k-bandwidth-NFA}
    Given $k\in\N$, an NFA $A = (Q, \Sigma, \delta, q_0, F)$ is \emph{$k$-bandwidth} ($k$-BW-NFA) if the states can be indexed as $Q = \{q_0,...,q_{n-1}\}$ such that for all $q_i,q_j\in Q$ and $\ta\in\Sigma\cup\{\varepsilon\}$,  $q_j\in \delta(q_i,\ta)$ implies $0 < (j-i) \bmod n \leq k $.
\end{definition}

This definition corresponds to the notion of distance bounded NFA from~\cite{ChoFSW25}, where they just have been briefly considered and not formally investigated in a lot of detail. As bandwidth is an established name for the underlying graph property \cite{ChinnCDG82}, we have chosen this name in this paper. Intuitively, the bandwidth $k$ defines a maximum distance that can be traversed by a single transition between states. This results in either a \emph{line of states} where transitions may \emph{jump} a certain number of states ahead, or it results in a \emph{circle of states} that can be traversed in a similar manner, but repeatingly.

\subsection{Constantly-bounded co-transcriptional splicing as a programming platform}

Motivated by the practical problems discussed in the introduction, we introduce the following problems. First, the most general one asks whether a given regular language represented by some NFA can be obtained with a given bandwidth $k$.

\begin{problem}[Linearization of NFA]\label{prob:lin-nfa}
    Given an NFA $A$ and a number $k \geq 1$, does there exist a $k$-BW-NFA $B$ such that $L(A) = L(B)$?
\end{problem}

In many cases, especially in the context of synthesizing a certain set of RNA sequences, we do not need the full power of NFA to express target/input languages. Hence, we consider the following subproblem as highly practically relevant.

\begin{problem}[Linerization of Finite Languages]\label{prob:lin-fin-lang}
    Given some finite language $L_f\subset\Sigma^*$ and a number $k \geq 1$, does there exist a $k$-BW-NFA $A$ such that $L_f = L(A)$? Does there exist a constant bound for $k$ such that all finite languages lie in the class of languages that can be expressed with $k$-BW-NFA?
\end{problem}

Up to this point, we ask basic existence questions but do not specify any constraints with respect to the size of the obtained automata. Taking the embedding in co-transcriptional splicing into account, minimizing the size of the automata and, thus, minimizing the necessary size of the encoding as a DNA sequence also remains of highly practical relevance, as shorter DNA strands promise more stability and reduced cost. Deciding the existence of some $n$ state $k$-BW-NFA $A$ for a target finite language $L_f$ has been shown to be NP-hard if in addition to $n$ and $L_f$, $k$ is also left as part of the input. Hence, we wondered whether this problem can be efficiently decided if $k$ is bound by a constant, effectively resulting in an FPT algorithm for the following problem.

\begin{problem}[Minimization of $k$-BW-NFA]\label{prob:min-bw-nfa}
    Let $k\geq 1$ be a constant. Given a finite language $L_f$ and a number $n \geq 1$, does there exist a $k$-BW-NFA $A$ with $n$ states such that $L(A) = L_f$?
\end{problem}

As mentioned in the introduction, for all of the highly relevant problems we obtained answers with respect to their computational complexity. The more general variants, as to be seen, remain open for future research. We continue with some general observations regarding the accepted languages by $k$-BW-NFA.

\section{The Hierarchy of $k$-Bandwidth NFAs}

We begin with a central result that the classes recognizable by $k$-bandwidth NFAs form a strict hierarchy based on the selection of $k$. Consider the following definition. 
Let $\mathcal{L}_k$ be the class of languages recognizable by $k$-BW-NFAs.

\begin{proposition}\label{proposition:proper-hierarchy}
    The proper inclusion $\mathcal{L}_k \subsetneq \mathcal{L}_{k+1}$ holds, for any $k$.
\end{proposition}
\begin{proof}
    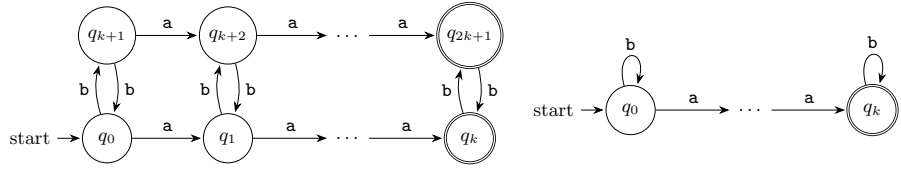
\begin{figure}[tb]
    \centering
    \begin{minipage}{0.5\linewidth}
    \scalebox{0.8}{\begin{tikzpicture}[shorten >=1pt, node distance=2cm, on grid, auto, >=Stealth]
        \node[state, initial] (q0) {$q_0$};
        \node[state] (q1) [right=of q0] {$q_1$};
        \node (dots1) [right=of q1] {$\dots$};
        \node[state, accepting] (qk) [right=of dots1] {$q_k$};
  
        \node[state] (qk1) [above=1.7cm of q0] {$q_{k+1}$};
        \node[state] (qk2) [right=of qk1] {$q_{k+2}$};
        \node (dots2) [right=of qk2] {$\dots$};
        \node[state, accepting] (q2k1) [right=of dots2] {$q_{2k+1}$};

        \path[->]
            (q0) edge node {$\ta$} (q1)
                 edge [bend left=15] node {$\tb$} (qk1)
            (q1) edge node {$\ta$} (dots1)
                 edge [bend left=15] node {$\tb$} (qk2)
            (dots1) edge node {$\ta$} (qk)
            (qk) edge [bend left=15] node {$\tb$} (q2k1)

            (qk1) edge node {$\ta$} (qk2)
                  edge [bend left=15] node {$\tb$} (q0)
            (qk2) edge node {$\ta$} (dots2)
                  edge [bend left=15] node {$\tb$} (q1)
            (dots2) edge node {$\ta$} (q2k1)
            (q2k1) edge [bend left=15] node {$\tb$} (qk);
    \end{tikzpicture}}
    \end{minipage}
    \begin{minipage}{0.05\linewidth}
        \ \\
    \end{minipage}
    \begin{minipage}{0.4\linewidth}
\scalebox{0.8}{\begin{tikzpicture}[shorten >=1pt, node distance=2cm, on grid, auto, >=Stealth]
        \node[state, initial] (q0) {$q_0$};
        \node (dots) [right=of q0] {$\dots$};
        \node[state, accepting] (qk) [right=of dots] {$q_k$};
        \path[->]
            (q0) edge [loop above] node {$\tb$} (q0)
                 edge node {$\ta$} (dots)
            (dots) edge node {$\ta$} (qk)
            (qk) edge [loop above] node {$\tb$} (qk);
    \end{tikzpicture}}
            
    \end{minipage}
    \caption{(Left) The construction of the $(k+1)$-BW-NFA accepting $\tb^*(\ta\tb^*)^k$. The double circles represents its accepting states.
    (Right) The NFA accepting $\tb^*(\ta\tb^*)^k$.
    }
    \label{fig:construction-with-and-without-selfloop}
    \end{figure}

    Consider the regular expression $\tb^*(\ta\tb^*)^k$.
    We will prove $\tb^*(\ta\tb^*)^k \in \mathcal{L}_{k+1} \setminus \mathcal{L}_k$.
    Firstly, we can see $\tb^*(\ta\tb^*)^k \in \mathcal{L}_{k+1}$ by Fig.~\ref{fig:construction-with-and-without-selfloop} (Left) which illustrates a $(k+1)$-BW-NFA accepting the language.
    Note that while the NFA illustrated in Fig.~\ref{fig:construction-with-and-without-selfloop} (Right) accepts the language, it is not a $(k+1)$-BW-NFA.
    Secondly, we will prove $\tb^*(\ta\tb^*)^k \not\in \mathcal{L}_{k}$ by contradiction.
    Assume $\tb^*(\ta\tb^*)^k \in \mathcal{L}_{k}$.
    Then, there exists a $k$-BW-NFA $A = (Q, \Sigma, \delta, q_0, F)$ with $Q = \{q_0, q_1, \ldots, q_{n-1}\}$ such that $L(A)=\tb^*(\ta\tb^*)^k$.
    Consider the sequence $w=\beta(\ta \beta)^k \in L(A)$ such that $\beta=\tb^{n+1}$.
    Then, while reading each occurrence of $\beta$, the automaton visits at least one state more than once.
    For each $\ell$ with $1 \leq \ell \leq k+1$,
    let $q_\ell'$ be a state that the automaton visits at least twice while reading the $\ell$-th occurrence of $\beta$.
    We define $C_\ell$ as the set of all states that the cycle (from $q_\ell'$ to $q_\ell'$) includes.
    
    Then, we will prove $C_i \cap C_j = \emptyset$ for any $i \neq j$ by contradiction.
    Assume $C_i \cap C_j \neq \emptyset$ for some $i \neq j$.
    Without loss of generality, we can assume $i<j$.
    Then, there exists a common state $q \in C_i \cap C_j$.
    Therefore, the automaton $A$ includes a loop from $q$ to $q_i$, then via an $\ta$-transition to $q_j$, and finally back to $q$.
    This loop includes the symbol $\ta$, which allows the automaton $A$ to accept words with more than $k$ occurrences of $\ta$,
    which contradicts the definition of $\tb^*(\ta\tb^*)^k$.
    Therefore, we have $C_i \cap C_j = \emptyset$ for any $i \neq j$.

    Let $Q_{x,y} \; (\subseteq Q)$ be the set of $y$ consecutive states starting from $q_x$, defined as $Q_{x,y}=\{ q_{(x+i) \bmod n} \mid i=0,\ldots,y-1 \}$.
    Let $\ell\in\{ 1,\ldots,k+1 \}$ be any integer.
    Then, since $A$ is a $k$-BW-NFA,
    for all $q_i,q_j\in C_\ell$ and $\ta\in\Sigma\cup\{\varepsilon\}$,  $q_j\in \delta(q_i,\ta)$ implies $0 < (j-i) \bmod n \leq k $.
    Therefore, for any integer $m\in\{ 1,2,\ldots \}$,
    the set $Q_{m,k}$ includes at least one element from $C_\ell$.
    Since $C_i \cap C_j = \emptyset$ holds for any $i \neq j$,
    the set $Q_{m,k}$ includes at least $k+1$ distinct states,
    which contradicts the fact $|Q_{m,k}|=k$.
    Therefore, $\tb^*(\ta\tb^*)^k \not\in \mathcal{L}_{k}$ holds.
\qed
\end{proof}

Hence, we know that for each fixed $k$, the class $\mathcal{L}_k$ is indeed a sub-regular language and a distinct class of languages compared to all $\mathcal{L}_{k'}$ where either $k' < k$ or $k' > k$. Related to the practical motivation for this model this implies that maximizing the achievable size of a removed factor by co-transcriptional splicing also strictly increases the languages that can be expressed using this operation.

    \section{Finite Languages and $k$-Bandwidth NFAs}

In this section, we consider the class of finite languages over $k$-BW-NFAs. From a practical viewpoint, this class of languages seems most interesting as a set of target RNA sequences encoded on DNA templates is typically finite.
Let $\mathcal{F}$ be the set of all finite languages. First, we can show for finite languages, a sort of trivial concatenation of words can be done to obtain a $k$-BW-NFA that represents this language. Indeed, even for a very small constant, $k = 2$, we can encode any finite language. The construction uses an alternating structure between two paths, one for reading the current word and the other one to skip ahead to a point after the current word is read.

\begin{lemma}\label{lemma:finite-existence-2BW}
    We have $\mathcal{F} \subseteq \mathcal{L}_2$.
\end{lemma}
\begin{proof}
    Let $\Sigma$ be some alphabet. 
    Let $L_f = \{w_1,w_2,...,w_m\} \subseteq \Sigma^*$ for some $m \ge 1$. 
    Let $k = \sum_{i=1}^{m}|w_i|$. 
    We construct a $2$-BW-NFA $A = (Q,\Sigma,\delta,q_0,F)$ s.t. $L(A) = L_f$. Set $Q = \{q_0,q_1,...,q_{2k-|w_m|+1}\}$ and define $\delta$ and $F$ by the following: For each $i\in[m]$, we distinguish between 2 cases. 

    First, if $m=1$, we do the following: We set $q_1\in\delta(q_0,w_1[1])$ as well as $q_2 \in \delta(q_0,\varepsilon)$. Then, for each $j\in[|w_1|-1]$, we add the transition $q_{2(j+1)-1}\in\delta(q_{2j-1},w_1[j+1])$ and the transition $q_{2(j+1)}\in\delta(q_{2j},\varepsilon)$. Finally, we define $q_{2|w_1|-1}\in F$. In total, we obtain two distinct outgoing paths from $q_0$, one reading $w_1$ and ending in a final state, the other reading $\varepsilon$ and ending in a non final state (see Fig.~\ref{fig:lemma:finite-existence-2BW-beginning}). 

    \begin{figure}[]
        \centering
    
            \resizebox{0.9\textwidth}{!}{
                \begin{tikzpicture}[>=Stealth, shorten >=1pt, auto, node distance=2.0cm]
                  \tikzset{
                    state/.style={circle, draw=black, thick, minimum size=8.8mm, inner sep=0pt, font=\small},
                    ystate/.style={state, fill=stateyellow},
                    tstate/.style={state, fill=stateteal},
                    vstate/.style={state, fill=stateviolet}
                  }
                
                  \node[vstate, initial, initial text=] (q0) {$q_0$};
                  \node[tstate, right=0.6cm of q0] (q1) {$q_1$};
                  \node[ystate, right=0.6cm of q1] (q2) {$q_2$};
                  \node[tstate, right=0.6cm of q2] (q3) {$q_3$};
                  \node[ystate, right=0.6cm of q3] (q4) {$q_4$};
                  \node[tstate, right=0.6cm of q4] (q5) {$q_5$};
                  \node[ystate, right=0.6cm of q5] (q6) {$q_6$};
                  \node[tstate, accepting, right=0.6cm of q6] (q7) {$q_7$};
                  \node[vstate, right=0.6 of q7] (q8) {$q_8$};
                
                  \path[->]
                    (q0) edge[bend left=45, edgeturquoise, very thick]
                      node[above] {$\ta$} (q1)
                    (q1) edge[bend left=45, edgeturquoise, very thick]
                      node[above] {$\tb$} (q3)
                    (q3) edge[bend left=45, edgeturquoise, very thick]
                      node[above] {$\tb$} (q5)
                    (q5) edge[bend left=45, edgeturquoise, very thick]
                      node[above] {$\ta$} (q7);
                
                  \draw[->, orange!80!yellow, very thick]
                    (q0) to[bend right=45] node[below] {$\varepsilon$} (q2);
                
                  \draw[->, orange!80!yellow, very thick]
                    (q2) to[bend right=45] node[below] {$\varepsilon$} (q4);
                
                  \draw[->, orange!80!yellow, very thick]
                    (q4) to[bend right=45] node[below] {$\varepsilon$} (q6);
                
                  \draw[->, orange!80!yellow, very thick]
                    (q6) to[bend right=45] node[below] {$\varepsilon$} (q8);
                \end{tikzpicture}
            }
        
        \caption{The initial part of the automaton that accepts the first word of $L_f$. The teal path correspondings to $w_1$ (here, for example $\mathtt{abba}$) while the orange path corresponds to the skipping $\varepsilon$-path that can be used to move to $q_8$.
        }
        \label{fig:lemma:finite-existence-2BW-beginning}
    \end{figure}
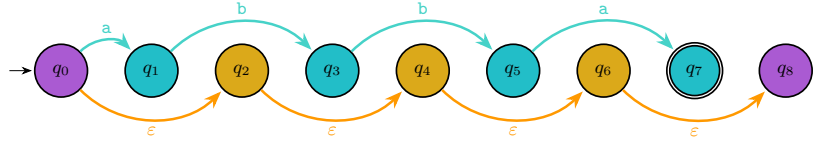
    
    Now, if $m \geq 2$, we inductively continue as follows. Let the largest state reached by the previous step be the state $q_{s}$, i.e., $s = \sum_{i'=0}^{i-1}2|w_{i'+1}|$. We reach $q_s$ only by the path that reads $\varepsilon$. Now we analogously construct two paths, one reading $w_i$ and the other reading $\varepsilon$. First, set $q_{s+1}\in\delta(q_{s},w_i[1])$ and $q_{s+2}\in\delta(q_{s},\varepsilon)$. Now, for each $j\in[|w_i|-1]$, we define $q_{s+2(j+1)-1} \in \delta(q_{s+2j-1},w_i[j+1])$ and $q_{s+2(j+1)} \in \delta(q_{s+2j},\varepsilon)$. Finally, set $q_{s+2|w_i|-1}\in F$. As before, we constructed two outgoing paths of the state $q_{s}$, one reading only $\varepsilon$ and ending in state $q_{s+2|w_i|}$, the other reading $w_i$ and ending in state $q_{s+2|w_i|-1}$ (see Fig.~\ref{fig:lemma:finite-existence-2BW-induction}). 

        \begin{figure}[]
        \centering
    
            \resizebox{0.75\textwidth}{!}{
                \begin{tikzpicture}[>=Stealth, shorten >=1pt, auto, node distance=1.9cm]
                  \tikzset{
                    state/.style={circle, draw=black, thick, minimum size=8.8mm, inner sep=0pt, font=\small},
                    ystate/.style={state, fill=stateyellow},
                    tstate/.style={state, fill=stateteal},
                    vstate/.style={state, fill=stateviolet}
                  }
                
                  \node[state] (q6) {$q_{s-2}$};
                  \node[state, accepting, right=0.6cm of q6] (q7) {$q_{s-1}$};
                
                  \node[vstate, right=0.6cm of q7] (qs) {$q_s$};
                  \node[tstate, right=0.6cm of qs] (qs1) {$q_{s+1}$};
                  \node[ystate, right=0.6cm of qs1] (qs2) {$q_{s+2}$};
                  \node[tstate, accepting, right=0.6cm of qs2] (qs3) {$q_{s+3}$};
                  \node[vstate, right=0.6cm of qs3] (qs4) {$q_{s+4}$};
                
                  \path[->]
                
                    (qs) edge[bend left=45, edgeturquoise, very thick]
                      node[above] {$c$} (qs1)
                    (qs1) edge[bend left=45, edgeturquoise, very thick]
                      node[above] {$d$} (qs3);
                
                  \draw[->, orange!80!yellow, very thick]
                    (q6) to[bend right=45] node[below] {$\varepsilon$} (qs);
                
                  \draw[->, orange!80!yellow, very thick]
                    (qs) to[bend right=45] node[below] {$\varepsilon$} (qs2);
                
                  \draw[->, orange!80!yellow, very thick]
                    (qs2) to[bend right=45] node[below] {$\varepsilon$} (qs4);
                
                \end{tikzpicture}
            }
        
        \caption{The inductive step of adding additional words to the automaton by the exemplary word $\mathtt{cd}$. Notice that $q_s$ could correspond to $q_8$ from the automaton in Fig.~\ref{fig:lemma:finite-existence-2BW-beginning}.}
        \label{fig:lemma:finite-existence-2BW-induction}
    \end{figure}
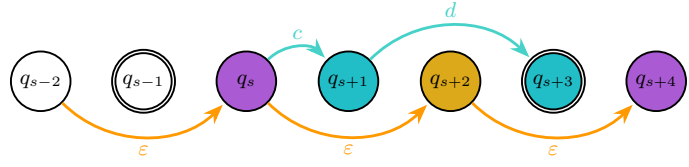

    For the last word $w_m$, we only have to construct one of the paths as the skipping $\varepsilon$ path is not necessary anymore, leading to $2k-|w_m|+1$ instead of $2k+1$ states (The $+1$ coming from the initial state $q_0$). Finally, we have to verify that $A$ actually has bandwidth $2$, i.e., each outgoing transition moves at most two states ahead. As all cases are constructed analogously, we only show this for the base step. First, there are transitions from $q_0$ to $q_1$ and from $q_0$ to $q_2$, which fulfill the constant. Next, for each $j\in[|w_1|-1]$, there is one transition from $q_{2j-1}$ to $q_{2(j+1)-1}$ and one transition from $q_{2j}$ to $q_{2(j+1)}$. For the first transition, we obtain a distance of $2(j+1)-1-(2j-1) = 2j+2-1-2j+1 = 2$. For the second transition, we obtain a distance of $2(j+1)-(2j) = 2j+2-2j=2$. Hence, for all transitions, $q_j\in\delta(q_i,a)$ implies $0 < (j-i) \bmod |Q| \leq 2 $. By the inductive construction, each and only the words $w_1$ to $w_m$ are accepted by $A$, resulting in $L_f = L(A)$. No transition skips over more than $2$ states and no state has more than $2$ outgoing transitions.
\qed
\end{proof}

\subsection{The Special Case of Bandwidth $1$}

It is rather apparent that $1$-BW-NFA do not have the same expressibility as $2$-BW-NFA. Not all finite languages can be encoded in a $1$-BW-NFA, even though $\varepsilon$-transitions are allowed. The main reasoning lies in the fact that we cannot skip ahead an arbitrary distance without implicitly allowing more combinations of letters to the read.

\begin{lemma}\label{lemma:finite-language-not-in-L1}
    We have $\mathcal{F} \not\subseteq \mathcal{L}_1$.
\end{lemma}
\begin{proof}
    We will show $\{\ta,\tb\tb\} \not\in \mathcal{L}_1$.
    Assume that there exists a 1-BW-NFA that accepts the language $\{\ta,\tb\tb\}$.
    Due to its linear structure, such an automaton contains exactly one transition labeled with $\ta$ and two transitions labeled with $\tb$ along a single path.
    Therefore, there exist words $w_1$, $w_2$, and $w_3$ such that the automaton accepts either $w_1 \ta w_2 \tb w_3$ or $w_1 \tb w_2 \ta w_3$.
    A contradiction.
\qed
\end{proof}

Nevertheless, they are of high interest as in the context of co-transcriptional splicing, a bandwidth of $1$ in an encoded automaton could result in only very small factors having to be spliced out (see, e.g., \cite{ChoFSW25_NC}). As each $1$-BW-NFA can be understood as a line of states in which some states might be set to final states, we can give a full characterization of all non-empty finite languages that can be expressed by some $1$-BW-NFA.

\begin{lemma}\label{lemma:characterization-1-BW-NFA-finite}
    Let $L_f$ be some finite language. There exists a $1$-BW-NFA $A$ with $n\geq 1$ states and $L(A) = L_f$ if and only if there exist languages $S_1, ... , S_{n-1}\subseteq \Sigma\cup\{\varepsilon\}$, for some $n \geq 1$, and a set $I \subseteq [n]$ such that
    $$L_f = \bigcup_{i\in I} S_1\cdot ... \cdot S_{i-1}.$$
    In the case of $i=1$, let the empty product be defined as $\{\varepsilon\}$.
\end{lemma}
\begin{proof}
    First, assume there exists a $1$-BW-NFA $A = (Q,\Sigma,\delta,q_1,F)$ with $n\geq 1$ states and $L(A) = L_f$. Assume w.l.o.g. that $Q = \{q_1,...,q_n\}$ such that for each $i\in[n-1]$ and each $a\in\Sigma\cup\{\varepsilon\}$ we have $\delta(q_i,a)\subseteq\{q_{i+1}\}$. For each $i\in[n-1]$, we define
    $ S_i = \{\ a\in\Sigma\cup\{\varepsilon\} \mid q_{i+1}\in\delta(q_i,a)\ \}.$
    Now let $t\in[n]$. A word $w\in\Sigma^*$ labels a path from $q_1$ to $q_t$ if and only if it can be written as
    $ w = w[1]w[2]\cdots w[t-1]$
    with $w[i]\in S_i$ for every $i\in[t-1]$ (recall that $w[i] = \varepsilon$ is allowed). Thus, the set of words on the path from $q_1$ to $q_t$ is exactly $S_1\cdot S_2 \cdot ... \cdot S_{t-1}$. By that,
    $$ L_f = L(A) = \bigcup_{q_t\in F} S_1\cdot ... \cdot S_{t-1}. $$
    So, setting $I=\{ i \mid q_i \in F \}$, the claim holds.

    For the other direction, assume there exist languages $S_1,...,S_{n-1}\subseteq \Sigma\cup\{\varepsilon\}$, for some $n\geq 1$, and a set $I\subseteq[n]$ such that
    $$ L_f = \bigcup_{i\in I} S_1\cdot ... \cdot S_{i-1}.$$
    We construct an NFA $A  = (Q,\Sigma,\delta,q_1,F)$ with states $Q = \{q_1,...,q_n\}$, initial state $q_1$, final states $F = \{\ q_i \mid i\in I\}$, and transitions
    $ q_{j+1} \in \delta(q_j,a) \text{ iff } a\in S_j,$
    for $j\in[n-1]$. This automaton clearly has bandwidth $1$. Following the construction, we observe the following, concluding this direction:
    $$L(A) = \bigcup_{i\in I} S_1\cdot ... \cdot S_{i-1} = L_f.$$ \qed
\end{proof}

So, by now we know that not all finite languages can be expressed by a $1$-BW-NFA but we know their general structure by this previous characterization. Given a set of words, however, it might not be immediately clear whether it can be structured with respect to the previous characterization. Hence, we consider the feasibility of Problem~\ref{prob:lin-fin-lang} for the case $k=1$. 

\begin{theorem}\label{theorem:existence-bw-nfa-k1}
Let \(L_f\) be a finite language over an alphabet \(\Sigma\), given explicitly
as a list of words, and let
\(
N=\max_{w\in L_f}|w|.
\)
We can decide in \(O(|L_f|^2N^4)\) time whether there exists a bandwidth-1 NFA
\(A\) such that
\(
L(A)=L_f.
\)
\end{theorem}

\begin{proof}
Let
\(
L_{\max}=\{w\in L_f \mid |w|=N\}.
\)
If \(L_f=\emptyset\), the problem is trivial. Hence assume \(L_f\neq\emptyset\).
As in Lemma~\ref{lemma:characterization-1-BW-NFA-finite}, a bandwidth-1 NFA accepting a finite language may be viewed
as a path
\[
q_0\to q_1\to\cdots\to q_N,
\]
where the transition from \(q_{i-1}\) to \(q_i\) has a label set
\(
S_i'\subseteq \Sigma\cup\{\varepsilon\}.
\)
Every word of length \(N\) must use a non-\(\varepsilon\) label on every one of
the \(N\) transitions. Hence the non-\(\varepsilon\) backbone is uniquely forced
by \(L_{\max}\). For \(i\in[N]\), define
\[
S_i=\{w[i] \mid w\in L_{\max}\}.
\]
A necessary condition is
\(
L_{\max}=S_1S_2\cdots S_N.
\)
If this condition fails, we reject, so we assume from now on that
\(
L_{\max}=S_1S_2\cdots S_N.
\)

\medskip
\noindent
\textbf{Phase 1: guessing nonsingleton \(\varepsilon\)-assignments.}
Let
\(
J=\{i\in[N] \mid |S_i|\ge 2\}.
\)
Since
\[
|L_{\max}|=\prod_{i=1}^N |S_i|,
\]
we have
\(
2^{|J|}\le |L_{\max}|\le |L_f|.
\)
Therefore
\(
|J|\le \log_2 |L_f|.
\)
Thus all possible choices of which transitions \(i\in J\) also receive an
\(\varepsilon\)-label can be enumerated in polynomial time in \(|L_f|\).

Fix one such choice \(E_J\subseteq J\). We now decide whether this choice can
be completed by assigning \(\varepsilon\)-labels to singleton transitions.

\medskip
\noindent
\textbf{Phase 2: right-end normal form for singleton blocks.}
We group the singleton transitions into maximal consecutive blocks with the
same singleton label. A singleton block has the form
\(
a^m
\)
for some \(a\in\Sigma\), and two consecutive singleton blocks have different
letters. Non-singleton transitions separate singleton blocks.

We use the following normal form. In each maximal singleton block, if \(d\)
transitions are assigned \(\varepsilon\), then these \(d\) transitions are placed at
the right end of the block.
Consider one maximal singleton block
\(
a^m.
\)
Suppose \(d\) of its transitions are assigned \(\varepsilon\). For any final
state after the end of the block, the contribution of this block depends only on
\(d\), not on the exact positions of the \(d\) optional transitions: it is
\[
\{a^{m-d},a^{m-d+1},\ldots,a^m\}.
\]

Now consider a final state inside the block, after \(h\) positions of this block.
If \(e\) of the first \(h\) positions are optional, then the block
contribution is
\[
\{a^{h-e},a^{h-e+1},\ldots,a^h\}.
\]
After moving all \(d\) optional positions to the right end of the block, this
same interval of exponents can be obtained by setting the corresponding
states in the block final. By the way we defined the language accepted with safe final states, every state
whose prefix language is contained in \(L_f\) will be final, and no
unsafe state will be, in other words, moving the optional transitions to the right end of the block cannot
introduce words outside \(L_f\), in other words, every word previously accepted inside the
block is still accepted at one of the corresponding safe endpoints. 
Hence it is enough to store, for each singleton
block \(B_r\), a number
\(
d_r
\)
of rightmost \(\varepsilon\)-transitions.

Let
\(
D=(d_1,\ldots,d_t)
\)
denote the current vector of singleton \(\varepsilon\)-counts, together with
the fixed non-singleton choice \(E_J\), where $t$ is the number of transition blocks consisting of identical labels. For \(0\le i\le N\), let
\(
P_i(D)
\)
be the language accepted by setting \(q_i\) final under the current $\varepsilon$-assignment $D$.

A state \(q_i\) is safe if
\(
P_i(D)\subseteq L_f.
\)
Define
\(
F_{\mathrm{safe}}(D)=\{q_i \mid P_i(D)\subseteq L_f\},
\)
and
\[
L(D)=\bigcup_{q_i\in F_{\mathrm{safe}}(D)}P_i(D).
\]
Thus \(L(D)\) is the largest language accepted by the current transition
assignment without accepting any word outside \(L_f\).

\medskip
\noindent
\textbf{Phase 3: forced inference on singleton blocks.}
Initialize \(D=(0,\ldots,0)\). 

At every step compute
\(
U=L_f\setminus L(D).
\)
If \(U=\emptyset\), the current assignment realizes \(L_f\), so we accept this
guess \(E_J\) and the current $D$.

Otherwise let
\(
\ell=\max\{|w| \mid w\in U\},
\)
and let
\(
U_\ell=\{w\in U \mid |w|=\ell\}.
\)

For a singleton block \(B_r\), let \(D+e_r\) denote the assignment obtained
from \(D\) by adding one more rightmost \(\varepsilon\)-transition to \(B_r\).
We say that \(r\) is a candidate for \(w\in U_\ell\) if
\[
w\in L(D+e_r)
\]
and the final state \(q_N\) remains safe, that is,
\[
P_N(D+e_r)\subseteq L_f.
\]
Let
\(
C_w=\{r \mid r\text{ is a candidate for }w\}.
\)
If \(C_w=\emptyset\) for some \(w\in U_\ell\), then this guess \(E_J\) cannot be
completed, and we reject this guess.

Otherwise, by the forcing claim below, if this guess admits a valid completion,
then there exists some
\(
w^\star\in U_\ell
\)
with
\(
|C_{w^\star}|=1.
\)
Let
\(
C_{w^\star}=\{r\}.
\)
We increment \(d_r\) by one and continue.

\medskip
\noindent
\textbf{Forcing claim for singleton blocks.}
We prove the claim used above. Suppose that the current assignment \(D\) has a valid completion
\(D^\star\ge D\), still under the fixed nonsingleton choice \(E_J\), such that
\[
L(D^\star)=L_f.
\]
We prove that there exists a word \(w^\star\in U_\ell\) with \(|C_{w^\star}|=1\).

Since \(U\neq\emptyset\) and \(D^\star\) realizes \(L_f\), some word of
\(U_\ell\) is accepted at a safe final state under \(D^\star\). Consider an
accepting computation of such a word, chosen so that the number of singleton
transitions that are skipped using an $\varepsilon$-transition under \(D^\star\) but not yet skippable under \(D\)
is minimal. At least one such new $\varepsilon$-transition (singleton skip) is used, otherwise the word
would already belong to \(L(D)\).
In particular, some singleton block contributes one new skip to the current highest
uncovered layer.

We now show that one may choose such a contribution with a private witness.
Fix one singleton block \(B_i\) whose next rightmost transition is skipped in
the completion. Freeze the choices of all other transitions as follows. For
singleton blocks, use the choices of the computation described above. For the
nonsingleton transitions, choose non-\(\varepsilon\) letters, subject only to the
following separator convention: whenever two singleton runs with the same letter $a$
would become adjacent after deleting a non-singleton block that has some letter $b\neq a$ (such a letter exists, since the block is not a singleton one), keep $b^k$ from the non-singleton block's accepted words, where $k$ is the length of the block. If the singleton letters on the two sides are
already different, no separator is needed.

With these choices fixed, the relevant part of the accepting path becomes an
ordinary word. By the separator convention, distinct singleton runs do not merge.
Thus its singleton-run skeleton has the form
\[
a_1^{m_1}a_2^{m_2}\cdots a_s^{m_s},
\qquad
m_j>0,\quad a_j\neq a_{j+1}.
\]
Let the chosen block \(B_i\) correspond to the run \(a_i^{m_i}\). Form the word
\(
w_i
\)
by decreasing this run by one and leaving all other fixed choices unchanged.

Then \(w_i\in L(D+e_i)\). Moreover,
\[
P_N(D+e_i)\subseteq P_N(D^\star)\subseteq L_f,
\]
so \(i\) is a candidate for \(w_i\).

The word \(w_i\) cannot be obtained by adding one more \(\varepsilon\)-transition
to any other singleton block. Indeed, if a different block \(B_j\) with
\(j<i\) produced \(w_i\), then the run \(a_i^{m_i}\) would have to start one
position too early and consume the last letter before it. This is impossible
because, by construction, that preceding letter is different from \(a_i\). The
case \(j>i\) is symmetric: then the run \(a_i^{m_i}\) would have to end one
position too late and consume the first letter after it, again impossible because
that letter is different from \(a_i\). The boundary cases are one-sided versions
of the same argument.

Hence no other singleton block is a candidate for \(w_i\). Also \(w_i\notin
L(D)\), since producing it requires exactly the new skip in \(B_i\), and no other
block can simulate that skip. Since the final state used by the completion is
safe, \(w_i\in L_f\). Therefore
\(
w_i\in U.
\)
By construction it belongs to the current highest uncovered layer, so
\(
w_i\in U_\ell.
\)
Thus
\(
C_{w_i}=\{i\}.
\)
This proves that whenever a valid completion exists, the algorithm finds a word
with a unique candidate singleton block. Therefore rejecting a branch only when
no such word exists is sound.$\blacktriangleleft$

\medskip
\noindent
\textbf{Verification and complexity.}
If the algorithm accepts, it has explicitly constructed a bandwidth-1 NFA
whose safe final states accept exactly \(L_f\). Conversely, suppose a valid
bandwidth-1 NFA exists. Choose the branch \(E_J\) corresponding to its
\(\varepsilon\)-choices on non-singleton transitions, and put its singleton
part into right-end normal form. By the forcing claim, at every step of Phase 3
the algorithm chooses a singleton block that is forced in every valid
completion of the current assignment. Therefore this branch never rejects and
eventually reaches an assignment \(D\) with
\(
L(D)=L_f.
\)
Thus the algorithm accepts exactly the bandwidth-1 finite languages.

Phase 1 enumerates at most \(|L_f|\) choices for the
\(\varepsilon\)-assignments on non-singleton transitions.

Fix one such choice. The singleton-inference loop performs at most \(N\)
iterations, since each iteration adds one new \(\varepsilon\)-transition to a
singleton block. In one iteration, there are at most \(N\) singleton blocks to
test. For each candidate increment \(D+e_r\), we recompute the prefix languages
\(P_i(D+e_r)\), the safe states, and the language \(L(D+e_r)\). Using the prefix tree
of \(L_f\), every product enumeration is stopped as soon as a word leaves the
tree. If the product remains safe, it contains at most \(|L_f|\) words. Therefore
one such recomputation costs
\(
O(|L_f|N^2).
\)
Thus one iteration costs
\(
O(N\cdot |L_f|N^2)=O(|L_f|N^3),
\)
and one fixed non-singleton guess costs
\(
O(N\cdot |L_f|N^3)=O(|L_f|N^4).
\)
Multiplying by the at most \(|L_f|\) non-singleton guesses gives the total running
time
\[
O(|L_f|^2N^4).
\]\qed
\end{proof}

This settles this part of the problem on a positive note with respect to practical feasibility, but since the running time depends polynomially on the number of words in the language, the algorithm does not work in polynomial time w.r.t. to an arbirary DFA (let alone NFA) input. This leaves open the question whether an efficient decision procedure exists when the input language is presented as a finite automaton rather than a list of words (as the latter can be exponentially larger than the former for the same language). For $k \geq 2$, we know that each finite language can be encoded and for $k=1$ we now know that we can decide in polynomial time whether a finite language given as a list of words can be encoded in a $k$-BW-NFA. For $k=1$, the size of the resulting automata is determined by the longest words in the given language~$F$. For $k \geq 2$, the size of the resulting automaton is not apparent. Hence, asking whether such an automaton with $n$ states can be obtained, remains of high practical interest.

\subsection{Minimization of $k$-BW-NFAs is Hard for fixed $k\geq 2$}

In this last part, we consider the computational complexity of the minimization problem for $k$-bandwidth NFAs. In \cite{ChoFSW25_NC}, it was shown for an analogous model that if $k$ is left as part of the input, the problem is generally NP-hard. As it turns out, even though its expressiveness is severely restricted if the parameter $k$ is chosen to be fixed as a small number, the problem of deciding the existence of an $n$ state $k$-bandwith NFA is also already NP-hard in the case of fixing $k=2$ and setting the required language to be a finite set of words.

\begin{theorem}\label{theorem:minization-hard-fixed-k}
    Let $L_f$ be a finite language and $n\in\mathbb{N}$. Deciding the existence of a $2$-BW-NFA $A$ with $n$ states such that $L(A) = L_f$ is NP-complete.
\end{theorem}
\begin{proof}
    Regarding NP-containment, indeed,
    we can define a polynomial time verifier for a certificate containing a $k$-BW-NFA with $n$ states. The maximum size of $A$ lies in $O(n\cdot k \cdot |\Sigma|)$ for the total number of states $n$, the permissible bandwidth $k$ that implies the number of states a transition can be added to, and the size of the alphabet $|\Sigma|$ that determines the maximum number of transitions between any two given states. To verify this solution, we can construct the minimal DFA representation of $L_f$ in polynomial time, construct its complement and then its intersection with $A$ in polynomial time. If the resulting automaton has a non-empty language, the verifier returns false, otherwise true.

    NP-hardness is shown by a reduction from the $\texttt{Biclique-Cover-Number}$ problem. We recall its definition. Let $G = (V,E)$ be a bipartite graph over vertices $V = V_1\cup V_2$ with w.l.o.g. $V_1 = \{a_1,a_2,...,a_\ell\}$ and $V_2 = \{b_1,b_2,...,b_r\}$, for some numbers $\ell,r\in\N$, such that for all edges $\{u,v\}\in E$ we have $u\in V_1$ and $v\in V_2$. A biclique-edge-cover $C = \{C_1,C_2,...\}$ of $G$ is a set of bicliques (i.e., a set of complete bipartite subgraphs of $G$) that cover all edges in $G$. For each $C_i\in C$, we assume that it is a subset of $V$, since these elements uniquely define all edges in a complete subgraph. Deciding whether there exists a biclique-edge-cover with at most $c$ bicliques, for some number $c > 1$, is known to be generally NP-hard~\cite{ORLIN1977406}. 

    Now, let $G$ be a bipartite graph given as above and let $c\in\N$ with $c > 1$ be some positive number. We proceed by constructing an instance of the above mentioned minimization problem. We set $\Sigma = V$, $L_f = \{\ a_ib_j\mid \{a_i,b_j\}\in E\ \}$, and $n = 4(c-1)$.

    ($\Rightarrow$): Assume that there exists a biclique-edge-cover $C = \{C_1,\dots,C_c\}$ of $G$ with $c>1$. We construct a $2$-BW-NFA $A$ with $n = 4(c-1)$ states.
    Here, we only sketch the idea (see also Figure~\ref{fig:proof:theorem:minization-hard-fixed-k-direction1} for the case $|C|=3$). For each biclique $C_i$, we use a gadget with four states $p_0 < p_1 < p_2 < p_3$. For every $\ta\in C_i \cap V_1$, we add a transition from $p_0$ to $p_1$ reading $\ta$, and for every $\tb\in C_i\cap V_2$, we add a transition from $p_1$ to $p_3$ reading $\tb$, where $p_3$ is a final state. Hence, this gadget accepts exactly the words $\ta\tb$ corresponding to edges covered by $C_i$. Moreover, we add an $\varepsilon$-transition from $p_0$ to $p_2$ and use a $\varepsilon$-transitions from $p_2$ to chain gadgets, so that the automaton may either continue to a later gadget or accept a word in the current one. To save four states overall, the last two bicliques $C_{c-1}$ and $C_c$ are represented by one combined gadget. It uses a common initial state $p_0$ and a common final state $p_3$ but two distinct middle branches for $C_{c-1}$ and $C_c$. In this way, the whole construction uses exactly $4(c-1)$ states and accepts precisely the language $L_f$. 
    
    \begin{figure}[h!]
    \centering

    \resizebox{1\textwidth}{!}{
        \begin{tikzpicture}[>=Stealth, shorten >=1pt, auto, node distance=1.95cm]
          \tikzset{
            state/.style={circle, draw=black, thick, minimum size=8.5mm, inner sep=0pt, font=\small},
            tealstate/.style={state, fill=stateteal},
            ystate/.style={state, fill=stateyellow},
            every loop/.style={min distance=10mm}
          }
        
          \node[tealstate, initial, initial text=] (q0) {$q_0$};
          \node[tealstate, right=0.8of q0] (q1) {$q_1$};
          \node[tealstate, right=0.8cm of q1] (q2) {$q_2$};
          \node[tealstate, accepting, right=0.8cm of q2] (q3) {$q_3$};
        
          \node[ystate, right=0.8cm of q3] (q4) {$q_4$};
          \node[ystate, right=0.8cm of q4] (q5) {$q_5$};
          \node[ystate, right=0.8cm of q5] (q6) {$q_6$};
          \node[ystate, accepting, right=0.8cm of q6] (q7) {$q_7$};
        
          \path[->]
            (q0) edge[bend left=35, red, very thick]
              node[above] {$C_1\cap V_1$} (q1)
            (q1) edge[bend left=35, blue, very thick]
              node[above] {$C_1\cap V_2$} (q3)
            (q0) edge[bend right=35, edgepurple, very thick]
              node[below] {$\varepsilon$} (q2)
            (q2) edge[bend right=35, edgepurple, very thick]
              node[below] {$\varepsilon$} (q4)
        
            (q4) edge[bend left=35, red, very thick]
              node[above] {$C_2\cap V_1$} (q5)
            (q4) edge[bend right=35, red, very thick]
              node[below] {$C_3\cap V_1$} (q6)
        
            (q5) edge[bend left=35, blue, very thick]
              node[above] {$C_2\cap V_2$} (q7)
            (q6) edge[bend right=35, blue, very thick]
              node[below] {$C_3\cap V_2$} (q7);
        \end{tikzpicture}
    }
    
    \caption{Construction of a $2$-BW-NFA $A$ for a biclique-edge-covering $C$ containing three bicliques. The light blue and orange colored states refer to two distinct gadgets.
    }
    \label{fig:proof:theorem:minization-hard-fixed-k-direction1}
\end{figure}
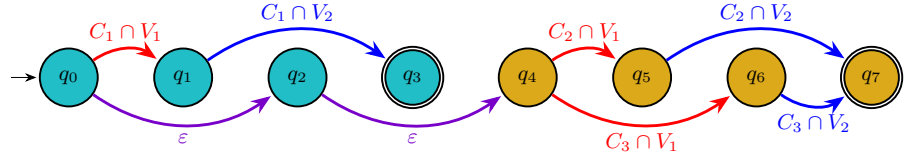

    ($\Leftarrow$): Now assume there exists a $2$-BW-NFA $A$ with at most $n = 4(c-1)$ states such that $L(A) = L_f$. It can be shown that the small bandwidth $2$ and maximum state space $n$ results in a set of so called maximum middle states $M$ that each can be used to represent one biclique. Each state $q\in M$ represents a state that sits in the middle between a set of vertices from $V_1$ and a set of vertices from $V_2$ that can form a biclique (in the example in Figure~\ref{fig:proof:theorem:minization-hard-fixed-k-direction1}, we would have $M = \{q_1,q_5,q_6\}$). It can be shown that each state $q\in M$ can be used to construct a biclique in $G$ by bringing together all letters that lead to $q$ with all the letters that go out from $q$ to a final state. By showing that $A$ with $n = 4(c-1)$ states can only have at most $|M| \leq c$ many maximum middle states, we obtain that such a $2$-BW-NFA $A$ emits a biclique-edge-covering $C$ of $G$ with at most $|C| = c$ bicliques. This direction of the proof, despite its simple intuition, requires a very comprehensive and detailed case distinction to actually confirm that $|M| > c$ is not possible. For verification purposes, the full detailed proof with the specific definition of $M$ can be found in the appendix of this paper\qed
\end{proof}

\section{Final Discussion}

We introduced the notion of $k$-bandwidth NFA and analyzed several of its properties. 
We proved a strict infinite hierarchy of language classes defined by bandwidth (see Fig.~\ref{fig:relationship}).

We showed that finite languages can be represented with maximum bandwidth $k = 2$. For $k = 1$, we obtained polynomial time decidability whether a $1$-BW-NFA $A$ with $L(A) = L_f$ exists for a given finite language $L_f$ presented as a list of words. The decision procedure implied a minimal size of such an automaton, allowing for a clear answer regarding the minimization problem. However, restricting $k$ to the next smallest constant $k=2$, we observe that minimization becomes NP-hard, establishing a clear boundary regarding feasibility.

This sets a solid theoretical foundation for the relevant questions of optimizing NFA that are to be encoded on circular DNA for RNA transcription. Some questions remain. First, the hardness result for the minimization problem for $k=2$ is based on the assumption that the alphabet is unbounded. It would be interesting to know whether this result can be extended to finite alphabets.

Next, most results were established for finite languages, as those are most relevant for our practical motivation. The hardness results clearly translate to unbounded regular languages. Nonetheless, it would be interesting to establish a characterization of all classes in the hierarchy of $k$-BW-NFA as well as to establish whether the decidability results persist in the case of non-finite regular languages.

Finally, this paper was concerned with exact matching of languages. In practice, covering the target language and allowing for certain unwanted words (given some restrictions) might still be a valid approach. Hence, considering $k$-BW-NFA in an approximate setting is also a worthwhile direction to consider in the future.

\subsection*{Acknowledgements.}

This work was supported in part by Institute of Information \& Communications Technology Planning \& Evaluation (IITP) under the Artificial Intelligence Convergence Innovation Human Resources Development (IITP-2025-RS-2023-00255968) grant funded by the Korea government (MSIT) to D.-J. C., 
JSPS KAKENHI Grant Number JP23K10976 to S.Z.F., 
Daiichi-Sankyo “Habataku” Support Program for the Next Generation of Researchers to D.~I., 
``Ambition Internationale'' de la Region d'Auvergne-Rh\^one-Alpes ``Rokam: RNA Origami Kinetic Assembly Model'' to S.~S. 
In addition, M.~W. gratefully acknowledges Dirk Nowotka (Kiel University) for his support and encouragement in participating in this collaboration. 

\bibliographystyle{splncs04}
\bibliography{main_arxiv}

\newpage

\appendix

\section{Full Proof of Theorem \ref{theorem:minization-hard-fixed-k}}\label{proof:theorem:minization-hard-fixed-k}

\subsection{NP-hardness: First Direction}\label{proof:theorem:minization-hard-fixed-k:direction1}

First, assume there exists a biclique-edge-cover $C = \{C_1,...,C_c\}$ with $c > 1$ many bicliques. We proceed by constructing a $2$-BW-NFA $A = (Q,\Sigma,\delta,q_0,F)$ with $n = 4(c-1)$ states $Q = \{q_0,...,q_{n-1}\}$. We proceed inductively through $C$. 

For the first biclique $C_1$, we consider the states $q_0$ to $q_3$. For each $a\in C_1\cap V_1$, add the transition $q_1 \in \delta(q_0,a)$. Next, for each $b\in C_1\cap V_2$, add the transition $q_3 \in \delta(q_1,b)$. Additionally, set $q_3\in F$ and add the $\varepsilon$-transition $q_2\in \delta(q_0,\varepsilon)$. The creates a gadget where we can either read words that represent the clique $C_1$ by taking a transition to $q_1$ from $q_0$, or we can decide to go to the state $q_2$ and read nothing for now. The state $q_2$ will be used to connect to all subsequent gadgets. See Figure~\ref{fig:1-proof:theorem:minization-hard-fixed-k} for a visualisation.

\begin{figure}[h!]
    \centering

    \resizebox{0.8\textwidth}{!}{
        \begin{tikzpicture}[>=Stealth, shorten >=1pt, auto]
            \tikzset{
                state/.style={circle, draw=black, thick, minimum size=9mm, inner sep=0pt, font=\small},
                ystate/.style={state, fill=stateyellow},
                every loop/.style={min distance=10mm}
            }
            
            \node[ystate, initial, initial text=] (q0) {$q_0$};
            \node[ystate, right=1.5cm of q0] (q1) {$q_1$};
            \node[ystate, right=1.5cm of q1] (q3) {$q_3$};
            \node[ystate, accepting, right=1.5cm of q3] (q4) {$q_4$};
        
            \path[->]
            (q0) edge[bend left=25, red, very thick]
              node[above] {$C_1 \cap V_1$} (q1)
            (q1) edge[bend left=25, blue, very thick]
              node[above] {$C_1 \cap V_2$} (q4)
            (q0) edge[bend right=25, edgepurple, very thick]
              node[below] {$\varepsilon$} (q3);
        \end{tikzpicture}
    }
    
    \caption{Gadget for the first biclique $C_1$.}
    \label{fig:1-proof:theorem:minization-hard-fixed-k}
\end{figure}
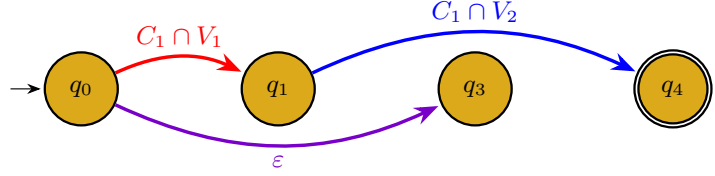

For all but the last two bicliques $C_i\in C$ with $i \leq c-2$, we proceed inductively like follows (see Figure~\ref{fig:2-proof:theorem:minization-hard-fixed-k} for a visualisation). Let $j = 4\cdot(i-1)$ be the number that refers to the initial state of the gadget for $C_i$. For example, if $i = 2$ is the second considered biclique, then $j = 4$ is the initial state of the second gadget, as only $q_0$ to $q_3$ have been considered before. First, we construct the path that represents the biclique $C_i$. For each $a\in C_i\cap V_1$, we add the transition $q_{j+1} \in \delta(q_j,a)$. Next, for each $b\in C_i\cap V_2$, we add the transition $q_{j+3} \in \delta(q_{j+1},b)$. Additionally, we set $q_{j+3}\in F$ to be a final state and the add the $\varepsilon$-transition $q_{j+2} \in \delta(q_j,\varepsilon)$. Up to now, we constructed the gadget analogously to the first one. What remains is to connect this gadget to the previous one. We add a $\varepsilon$-transition $q_j \in \delta(q_{j-1},\varepsilon)$ which connects to the $\varepsilon$ path constructed by the previous gadgets up to the initial state $q_0$.

\begin{figure}[h!]
    \centering

    \resizebox{1\textwidth}{!}{
        \begin{tikzpicture}[>=Stealth, shorten >=1pt, auto]
          \tikzset{
            state/.style={circle, draw=black, thick, minimum size=8.5mm, inner sep=0pt, font=\small},
            ystate/.style={state, fill=stateyellow},
            every loop/.style={min distance=10mm}
          }
        
          \node[state] (jm4) {$q_{j-4}$};
          \node[state, right=0.8 of jm4] (jm3) {$q_{j-3}$};
          \node[state, right=0.8cm of jm3] (jm2) {$q_{j-2}$};
          \node[state, accepting, right=0.8cm of jm2] (jm1) {$q_{j-1}$};
        
          \node[ystate, right=0.8cm of jm1] (j) {$q_j$};
          \node[ystate, right=0.8cm of j] (jp1) {$q_{j+1}$};
          \node[ystate, right=0.8cm of jp1] (jp2) {$q_{j+2}$};
          \node[ystate, accepting, right=0.8cm of jp2] (jp3) {$q_{j+3}$};
        
          \path[->]
            (jm4) edge[bend left=35, black, thick]
              node[above] {$C_{i-1}\cap V_1$} (jm3)
            (jm3) edge[bend left=35, black, thick]
              node[above] {$C_{i-1}\cap V_2$} (jm1)
            (jm4) edge[bend right=35, black, thick] 
              node[below] {$\varepsilon$} (jm2)
            (jm2) edge[bend right=35, edgepurple, very thick]
              node[below] {$\varepsilon$} (j)
            (j) edge[bend left=35, red, very thick]
              node[above] {$C_i\cap V_1$} (jp1)
            (j) edge[bend right=35, edgepurple, very thick]
              node[below] {$\varepsilon$} (jp2)
            (jp1) edge[bend left=35, blue, very thick]
              node[above] {$C_i\cap V_2$} (jp3);
        \end{tikzpicture}
    }
    
    \caption{Gadget for a middle biclique $C_i$.}
    \label{fig:2-proof:theorem:minization-hard-fixed-k}
\end{figure}
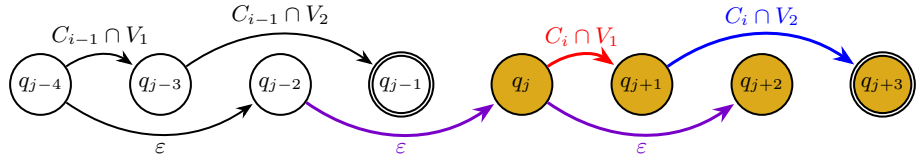

Finally, for the last two bicliques $C_{c-1}$ and $C_{c}$, we construct a combined gadget that also only uses $4$ states (see Figure~\ref{fig:3-proof:theorem:minization-hard-fixed-k} for a visualisation). Let $j = 4\cdot(c-2)$ be the numbers that refers to the initial state of the final gadget. For example, if $c = 3$, then $j = 4(3-2) = 4$. For all letters $a\in C_{c-1}\cap V_1$, we add a transition $q_{j+1} \in \delta(q_j,a)$, and for all letters $a'\in C_c\cap V_1$, we add a transition $q_{j+2}\in\delta(q_j,a)$. Next, for all letters $b\in C_{c-1}\cap V_2$, we add a transition $q_{j+3}\in\delta(q_{j+1},b)$, and for all letters $b'\in C_{c}\cap V_2$, we add a transition $q_{j+3}\in\delta(q_{j+2},b')$. Finally, we set $q_{j+3}\in F$ to be a final state and if $c > 2$, we add the connecting $\varepsilon$-transition to the previous gadget $q_i \in \delta(q_{j-2},\varepsilon)$.

\begin{figure}[h!]
    \centering

    \resizebox{1\textwidth}{!}{
        \begin{tikzpicture}[>=Stealth, shorten >=1pt, auto, node distance=1.95cm]
          \tikzset{
            state/.style={circle, draw=black, thick, minimum size=8.5mm, inner sep=0pt, font=\small},
            ystate/.style={state, fill=stateyellow},
            every loop/.style={min distance=10mm}
          }
        
          \node[state] (jm4) {$q_{j-4}$};
          \node[state, right=0.8cm of jm4] (jm3) {$q_{j-3}$};
          \node[state, right=0.8cm of jm3] (jm2) {$q_{j-2}$};
          \node[state, accepting, right=0.8cm of jm2] (jm1) {$q_{j-1}$};
        
          \node[ystate, right=0.8cm of jm1] (j) {$q_j$};
          \node[ystate, right=0.8cm of j] (jp1) {$q_{j+1}$};
          \node[ystate, right=0.8cm of jp1] (jp2) {$q_{j+2}$};
          \node[ystate, accepting, right=0.8cm of jp2] (jp3) {$q_{j+3}$};
        
          \path[->]
            (jm4) edge[bend left=35, black, thick]
              node[above=] {$C_{c-2}\cap V_1$} (jm3)
            (jm3) edge[bend left=35, black, thick]
              node[above] {$C_{c-2}\cap V_2$} (jm1)
            (jm4) edge[bend right=35, black, thick] 
              node[below] {$\varepsilon$} (jm2)
            (jm2) edge[bend right=35, edgepurple, thick]
              node[below] {$\varepsilon$} (j)
        
            (j) edge[bend left=35, red, very thick]
              node[above] {$C_{c-1}\cap V_1$} (jp1)
            (j) edge[bend right=35, red, very thick]
              node[below] {$C_c\cap V_1$} (jp2)
        
            (jp1) edge[bend left=35, blue, very thick]
              node[above] {$C_{c-1}\cap V_2$} (jp3)
            (jp2) edge[bend right=35, blue, very thick]
              node[below] {$C_c\cap V_2$} (jp3);
        \end{tikzpicture}
    }
    
    \caption{Gadget for the final gadget for $C_{c-1}$ and $C_{c}$.}
    \label{fig:3-proof:theorem:minization-hard-fixed-k}
\end{figure}
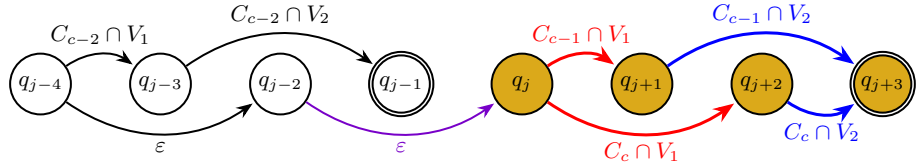

Now we have to show that $L(A) = L_f$. Let $u\in L_f$. Then $u = ab$ for some $a\in V_1$ and $b\in V_2$ such that $\{a,b\}\in E$. By the biclique-edge-covering $C$ we know that $a,b\in C_i$ for some $i\in[c]$. We can read the word $\varepsilon$ from $q_0$ to the state $q_{4(i-2)}$. As $c \geq 2$, this means we either go to the beginning of some gadget created further on or we stay in $q_0$. Then, from there on, we can read the word $ab$ and reach a final state. For the other direction, assume $u\in L(A)$ with $u = ab$, $a\in V_1$, $b\in V_2$. By the construction, for all states that have an incoming transition labeled with a letter from $V_1$, we know that we have to read some letter from $V_2$ next and reach a final state. We know that we can only read these two transitions if there exists some biclique $C_i$ with the letters $a$ and $b$ in it. Therefore, there exists an edge $\{a,b\}\in E$ and therefore $ab= u \in L_f$. This concludes this direction of the proof. \qed

\subsection{NP-hardness: Second Direction}\label{proof:theorem:minization-hard-fixed-k:direction2}

Assume there exists a $2$-BW-NFA $A = \{Q,\Sigma,\delta,q_0,F\}$ such that $L(A) = L_f$. We note some facts about $A$ with respect to the accepted language $L_f$. As all words are of length $2$, no outgoing path after a from the initial state $q_0$ reachable final state $q_f\in F$ that leads to any other final state can read anything else than the empty word $\varepsilon$. Also, we do not have to consider states that are not reachable from $q_0$. Finally, we do not have to consider states have have no outgoing path to any final state. Clearly, removing unreachable states, outgoing $\varepsilon$ paths after final states, and states that lead to no final states, only reduces the state space. So, we can assume $A$ to be trimmed in a sense that no unreachable states exist in $Q$, that all final states have no outgoing transitions, and that no states with no outgoing paths to any final states remain in $A$. This does not change the accepted language and for all NFA that are not trimmed that way we can obtain a trimmed version in polynomial time.

For the remainder of this subsection, we define the following notions for easier readability and handling of a lot ob sub cases. For two states $p,q\in Q$ with $q \neq p$, we write $p \rightarrow q$ iff there exists a transition from $p$ to $q$ in $A$, i.e., such that $q\in\delta(p,a)$, for some $a\in\Sigma\cup\{\varepsilon\}$. We write $p \rightarrow_* q$ if there exists some path from $p$ to $q$ in $A$, i.e., there exists some $u\in\Sigma^*$ such that $q\in\hat{\delta}(q,u)$. Extending this notion to specific letters to words, we say $p \overset{a}{\rightarrow} q$ iff there is a transition labeled with $a$ between $p$ and $q$, i.e., there exists a letter $a\in\Sigma\cup\{\varepsilon\}$ such that $q\in\delta(p,u)$. Analogously, we write $p \overset{u}{\rightarrow}_* q$ iff there exists a word $u\in\Sigma^*$ such that $q \in \hat{\delta}(p,u)$. Additionally, for the remainder of this subsection, we define the following sets. Let 
$$\mathtt{In}(q) = \{\ u\in\Sigma^* \mid q\in\hat{\delta}(q_0,u)\ \}$$
be the set of all words that can reach the state $q$, starting from the initial state $q_0$. On the other hand, we define
$$\mathtt{Out}(q) = \{\ v\in\Sigma^* \mid q_f\in\hat{\delta}(q,v) \text{ and } q_f\in\ F\ \}$$
to be the set of all words that can rach some final state $q_f\in F$, starting from the state $q\in Q$. 

For $A$, we define a specific subset $M\subset Q$ of states that fulfill some very specific properties. In principle, they are the rightmost middle states from which, locally, the largest possible combination of input and output words (or rather for this proof, letters) can be read to reach a final state. As we will see, from $M$ we obtain the set of corresponding bicliques we can find in the graph $G$. Let $M\subset Q$ be defined as the maximum set of states $q\in Q$ with the following properties:
\begin{enumerate}
    \item The state $q$ is part of an accepting path, i.e., there exist some words $u,v\in\Sigma^*$ such that $q_0 \overset{u}{\rightarrow_*} q \overset{v}{\rightarrow_*} q_f$, for some final state $q_f\in F$,
    \item for all $u\in\Sigma^*$ with $q_0 \overset{u}{\rightarrow_*} q$ we have $u\in V_2$,
    \item for all $v\in\Sigma^*$ with $q \overset{v}{\rightarrow_*} q_f$, for some final state $q_f\in F$, we have $v\in V_2$,
    \item there exists no other state $q'\in M$ such that $q \rightarrow_* q'$ and for which 
    $$\mathtt{In}(q) \subset \mathtt{In}(q') \text { and } \mathtt{Out}(q) = \mathtt{Out}(q'),$$
    \item there exists no other state $q'\in M$ such that $q' \rightarrow_* q$ and for which 
    $$\mathtt{In}(q) = \mathtt{In}(q') \text { and } \mathtt{Out}(q) \subset \mathtt{Out}(q'),$$
    \item and there exists no other state $q'\in M$ such that $q \rightarrow_* q'$ and for which
    $$\mathtt{In}(q) = \mathtt{In}(p) \text{ and } \mathtt{Out}(q) = \mathtt{Out}(p).$$
\end{enumerate}

First of all, we confirm that form $M$ we actually obtain a biclique covering in $G$.

\begin{lemma}\label{proof:theorem:minization-hard-fixed-k:lemma:M-to-Bicliques}
    If $A$ emits such a set $M$, there exists a biclique-covering of $G$ with $|M|$ bicliques.
\end{lemma}
\begin{proof}
    First of all, for each $q\in M$, we define the set $C_q = \mathtt{In}(q) \cup \mathtt{Out}(q)$. For each pair of letters $a,b\subset C_q$, with $a\in V_1$ and $b\in V_2$, we know that $ab\in L(A)$ by the definition of $M$. Hence, there exists an edge $\{a,b\}\in E$. By that, we see that $C_q$ is a biclique in $G$. Let $C$ be the set of all bicliques obtained that way from $M$. Suppose, $C$ is not a biclique-edge-cover of $G$. Then, there exists some edge $\{a,b\}\in E$ not covered by $C$. Thus, there exists no biclique $C_q$ for which $a,b\in C_q$. As $\{a,b\}\in E$, we have $ab\in L_f$ and therefore $ab\in L(A)$ by definition. So, there exists some state $q\in Q$ with $q_0 \overset{a}{\rightarrow_*} q \overset{b}{\rightarrow_*} q_f$, for some final state $q_f\in F$. Take the last state $p\in Q$ reachable from $q$ with the same input and output sets, i.e., with $q \rightarrow_* p$ with $\mathtt{In}(q) = \mathtt{In}(p)$ and $\mathtt{Out}(q) = \mathtt{Out}(p)$. If such a state exists, we know $p\in M$ by definition of $M$. Otherwise, we know $q\in M$ by definition of $M$, as $M$ is chosen maximal. But then, the biclique $C_q \in C$ and by that $\{a,b\}\in E$ is covered by $C$, which is a contradiction. So, $C$ must be a biclique-edge-covering of $G$.\qed
\end{proof}

Hence, if $|M| \leq c$, then this direction is done. The remainder of this subsection shows that $|M| \leq c$ must indeed hold if $|Q| \leq 4(c-1)$. Suppose $|M| > c$. Assume w.l.o.g. that the states $Q = \{q_0, q_1, ... , q_{|Q-1|}\}$ are labeled with regards to the total order $<_Q$ imposed on $Q$ by $A$ (i.e., $q_0 < q_1 < q_2 < .... < q_{|Q-1|}$). As $|M| > c$ and $|Q| < 4(c-1)$, there must exist two pairs of states $q_{i_1}$, $q_{j_1}$ and $q_{i_2}$, $q_{j_2}$ with $i_1 < j_1 \leq i_2 < j_2$ such that
$$j_1 - i_1 < 4 \text{ and } j_2 - i_2 < 4$$
by the pigeonhole principle. Assume w.l.o.g. that $q_{i_1}$ and $q_{j_1}$ is the first pair of states that fulfill $j_1 - i_1 < 4$. Indeed, we can use the fact that these two states are so close to each other to show that after $j_1$, no other state $p$ with $q_{j_1} < p$ can exist, contradicting the existence of at least the second state of the pair $q_{i_2}$ and $q_{j_2}$. 

Suppose there exists some state $q_\ell\in M$ with $q_{j_1} < q_\ell$, i.e., $j_1 < \ell$. Notice that due to bandwidth $2$, all subsequent states after $q_{j_1}$ must be reachable from $q_{j_1}$ or $q_{j_1-1}$. It follows a comprehensive case distinction based on the possible distances of $q_{i_1}$ and $q_{j_1}$ and the fact whether $q_\ell$ is reachable only from $q_{j_1}$, only from $q_{j_1-1}$ or reachable from both. For readability purposes, in the following, we rename $i_1$ to $i$ and $j_1$ to $j$.

\begin{itemize}

\item[(1)] First, we can exclude one distance independently of the fact from which states $q_\ell$ is reachable. Suppose $j = i+2$, i.e., assume $q_i$ and $q_j$ only have one state between them. The following cases emerge:
\begin{itemize}
    \item[(1.1)] Suppose $q_i \rightarrow q_j$ but $q_i \not\rightarrow q_{i+1}$, i.e., that $q_i$ has a transition to $q_j$ but no transition to $q_{i+1}$. As $q_{i},q_{j}\in M$, we can only have $q_{i} \overset{\varepsilon}{\rightarrow} q_j$. But then, $\mathtt{In}(q_i) \subseteq \mathtt{In}(q_j)$, $\mathtt{Out}(q_i) = \mathtt{Out}(q_j)$ and $q_i \rightarrow_* q_j$. If $\mathtt{In}(q_i) \subset \mathtt{In}(q_j)$, then this is a contradiction to $q_i\in M$. If $\mathtt{In}(q_i) = \mathtt{In}(q_j)$, then this also is a contradiction to $q_i\in M$.
    \item[(1.2)] Suppose $q_i \overset{u}{\rightarrow} q_{i+1}$, for some $u\in\Sigma^*$ but $q_i \not\rightarrow q_j$, i.e., that $q_i$ has a transition to $q_{i+1}$ reading $u$, but no transition to $q_j$. Then, due to bandwidth $2$, there must exist a transition $q_{i+1} \overset{v}{\rightarrow} q_j$ reaching some $v\in\Sigma^*$. As this connects $q_i$ and $q_j$ via $q_{i+1}$, and as $q_i,q_j\in M$, we must have $u = v = \varepsilon$. But then, $\mathtt{q_i} \subseteq \mathtt{q_{i+1}}$, $\mathtt{Out}(q_i) = \mathtt{Out}(q_{i+1})$, and as $q_i \rightarrow q_{i+1}$, we cannot have $q_i\in M$.
    \item[(1.3)] So, suppose there exists both, $q_i \overset{u}{\rightarrow} q_{i+1}$ as well as $q_i \overset{v}{\rightarrow} q_j$. As settled before, as $q_i,q_j\in M$, we must have $v = \varepsilon$. Two further subcases emerge.
    \begin{itemize}
        \item[(1.3.1)] Suppose also $q_{i+1} \overset{w}{\rightarrow} q_j$. Then, we must have $u = w = \varepsilon$ due to $q_i,q_j\in M$. But then, again, $\mathtt{In}(q_i) \subseteq \mathtt{In}(q_{i+1})$ and $\mathtt{Out}(q_i) = \mathtt{Out}(q_{i+1})$. A contradiction to $q_i\in M$ as $q_i \rightarrow q_{i+1}$.
        \item[(1.3.2)] So, $q_{i+1} \not\rightarrow q_j$. Suppose there exists some $b\in\mathtt{Out}(q_i)$ with $v\notin\mathtt{Out}(q_j)$. As $q_i \overset{\varepsilon}{\rightarrow} q_j$ is the only incoming transition in $q_j$, due to bandwidth $2$, we have $\mathtt{In}(q_i) = \mathtt{In}(q_j)$. But then, $q_j\notin M$ as $\mathtt{Out}(q_j) \subset \mathtt{Out}(q_i)$. A contradiction. So $\mathtt{Out}(q_i) = \mathtt{Out}(q_j)$. But then, $q_i\notin M$, which is also a contradiction. So (iii) cannot hold. As (i), (ii), and (iii) cannot hold, (1) cannot be true, concluding this main case.
    \end{itemize}
\end{itemize}
Now, we proceed with a case distinction regarding the fact whether $q_\ell$ is reachable only from $q_j$, only from $q_{j-1}$, or from both, and the remaining distances $j = i+1$ and $j = i+3$.

\item[(2)] Assume that $q_\ell$ is only reachable from $q_j$ and not from $q_{j-1}$, i.e., $q_j \rightarrow^* q_\ell$ and $q_{j-1} \not\rightarrow_* q_\ell$. Then $q_j$ functions as a bottleneck of letters from $V_1$ read on all paths that lead to $q_\ell$. Hence $\mathtt{In}(q_j) = \mathtt{In}(q_\ell)$ and $\mathtt{Out}(q_\ell) \subseteq \mathtt{q_j}$. If $\mathtt{Out}(q_\ell) \subset \mathtt{q_j}$, this is a contradiction to $q_\ell\in M$. If $\mathtt{Out}(q_\ell) = \mathtt{q_j}$, we have a contradiction to $q_j\in M$. This concludes this case.

\item[(3)] Assume that $q_\ell$ is only reachable from $q_{j-1}$ and not from $q_j$, i.e., $q_{j-1} \rightarrow_* q_\ell$ and $q_j \not\rightarrow_* q_\ell$.
\begin{itemize}
    \item[(3.1)] If $j = i+1$, i.e., the states $q_i$ and $q_j$ follow immediately after one another. Then $q_{j-1} = q_{i}$ and therefore $q_{j-1}\in M$. Due to bandwidth $2$, we obtain the same reasoning as in \textbf{(2)} analogously. Hence, a contradiction.
    \item[(3.2)] If $j = i+3$, i.e., there are two other states between $q_i$ and $q_j$. As $q_{j-1} \rightarrow_* q_\ell$ and $q_j \not\rightarrow_* q_\ell$, there must exist a transition $q_{j-2} \rightarrow q_j$ due to the fact that $q_j$ must be reachable. The following cases emerge:
    \begin{itemize}
        \item[(3.2.1)] Suppose $q_i \rightarrow_* q_{j-1}$ and $q_i \not\rightarrow_* q_{i+1} = q_{j-2}$. Then, we can extend the argument from (2) or (3.1) analogously to the states $q_i$ and $q_{i+1} = q_{j-2}$. Hence, we obtain a contradiction here as well.
        \item[(3.2.2)] Suppose $q_i \rightarrow_* q_{i+1}=q_{j-2}$ and $q_i \not\rightarrow_* q_{i+2} = q_{j-1}$. Then $q_i$ cannot have any outgoing transition which is a contradiction.
        \item[(3.2.3)] So, assume $q_i \rightarrow_* q_{i+2} = q_{j-1}$ and $q_i \rightarrow_* q_{i+1} = q_{j-2}$. We obtain the following possible combinations.
        \begin{itemize}
            \item[(3.2.3.1)] Suppose $q_i \overset{u}{\rightarrow} q_{i+1}\overset{v}{\rightarrow} q_{i+2}$ and $q_i\not\rightarrow q_{i+2}$. As $q_{i+1} \rightarrow_* q_j$ and $q_j,q_i\in M$, we must have $u = \varepsilon$. But then $\mathtt{In}(q_i) \subseteq \mathtt{In}(q_{i+1})$ and $\mathtt{Out}(q_i) = \mathtt{Out}(q_{i+1})$. This is a contradiction to $q_i\in M$.
            \item[(3.2.3.2)] So, suppose $q_i \overset{u}{\rightarrow} q_{i+1}\overset{v}{\rightarrow} q_{i+2}$ and $q_i\overset{w}{\rightarrow} q_{i+2}$. As $q_\ell \in M$, $q_{i+2} = q_{j-1} \rightarrow_* q_\ell$ and as $q_j\in M$ and $q_{i+1} \rightarrow_* q_j$, we must have $u = v = w = \varepsilon$, as $q_i \in M$. Hence, we observe, again, that $\mathtt{In}(q_i) \subseteq \mathtt{In}(q_{i+1})$, $\mathtt{Out}(q_i) = \mathtt{Out}(q_{i+1})$, and $q_i \rightarrow q_{i+1}$. This is a contradiction to $q_i\in M$.
            \item[(3.2.3.3)] So, finally, suppose $q_i\overset{u}{\rightarrow} q_{i+2}$ and $q_{i+1}\overset{v}{\rightarrow} q_{i+2}$ but $q_i \not\rightarrow q_{i+1}$. Then, immediately, $\mathtt{In}(q_i) \subseteq \mathtt{In}(q_{i+2})$ and $\mathtt{In}(q_i) = \mathtt{In}(q_{i+1})$. This is a contradiction to $q_i\in M$ as $q_i \rightarrow q_{i+2}$. Hence, the case (3.2.3) cannot occur. But then, together with (3.2.1) and (3.2.2), we obtain that (3.2) cannot be true, concluding this main case.
        \end{itemize}
    \end{itemize}
\end{itemize}

Combining the results from (3.1), (3.2) and (1), we conclude that (3) cannot be true, concluding this main case.

\item[(4)] Hence, finally, suppose that $q_\ell$ is reachable from both $q_j$ and $q_{j-1}$, i.e., $q_{j-1} \rightarrow_* q_\ell$ and $q_{j} \rightarrow_* q_\ell$.

\begin{itemize}
    \item[(4.1)] First, suppose $j = i+1$, i.e. $q_i$ and $q_j$ immediately follow after one another. The following possible cases emerge. Notice that here we have to take the two following states after $q_j$ into account, i.e., $q_{j+1}$ and $q_{j+2}$.
    \begin{itemize}
        \item[(4.1.1)] Suppose $q_i \overset{u}{\rightarrow} q_j$ and $q_i \not\rightarrow q_{j+1}$. Then $u = \varepsilon$ as $q_i,q_j\in M$. But then, $\mathtt{In}(q_i) \subseteq \mathtt{In}(q_J)$ and $\mathtt{Out}(q_i) = \mathtt{Out}(q_J)$. This is a contradiction to $q_i\in M$.
        \item[(4.1.2)] Analogously, suppose $q_i \overset{u}{\rightarrow} q_{j+1}$ and $q_i \not\rightarrow q_j$. Then either $q_{j+1} = q_\ell$ or $q_{j+1}\rightarrow_* q_\ell$. No matter what, $u = \varepsilon$, as $q_i,q_\ell\in M$. But then $\mathtt{In}(q_i)\subseteq\mathtt{In}(q_{j+1})$ and $\mathtt{Out}(q_i) = \mathtt{Out}(q_{j+1})$. This is a contradiction to $q_i\in M$ as $q_i \rightarrow q_{j+1}$. So, by now, we must have $q_i \overset{u}{\rightarrow} q_j$ and $q_i \overset{u}{\rightarrow} q_{j+1}$.
        \item[(4.1.3)] Suppose $q_j \overset{u}{\rightarrow} q_{j+1}$ but $q_j \not\rightarrow q_{j+2}$. Then either $q_{j+1} = q_\ell$ or $q_{}j+1 \rightarrow_* q_\ell$. No matter what, analogously to (4.1.2) we obtain that neither $q_i$ or $q_j$ can be in $M$. A contradiction.
        \item[(4.1.4)] Suppose $q_j \overset{u}{\rightarrow} q_{j+2}$ but $q_j \not\rightarrow q_{j+1}$. Then $q_{j+1}$ only has an incoming transition from $q_i$ which has to be an $\varepsilon$ transition. The arguments as in (4.1.2) hold and we obtain a contradiction.
        \item[(4.1.5)] So, we must also have $q_j \overset{u}{\rightarrow} q_{j+1}$ and $q_j \overset{v}{\rightarrow} q_{j+2}$ together with  $q_i \overset{x}{\rightarrow} q_j$ and $q_i \overset{y}{\rightarrow} q_{j+1}$. First, as either $q_{j+1}$ or a later state must be $q_\ell$, we have $u = x = y = \varepsilon$. We observe that $\mathtt{In}(q_i) \subseteq \mathtt{In}(q_j)$ and $\mathtt{In}(q_j) = \mathtt{q_{j+1}}$. We see that $\mathtt{Out}(q_i) = \mathtt{Out}(q_j) \cup \mathtt{Out}(q_{j+1})$. We also see that $\mathtt{Out}(q_j) \supseteq \mathtt{Out}(q_{j+1}) \cup \mathtt{Out}(q_{j+2})$. Hence, plugging that in, we obtain 
        \begin{align*}
            \mathtt{Out}(q_i) &= \mathtt{Out}(q_j) \cup \mathtt{Out}(q_{j+1}) \\
                              & \supseteq \mathtt{Out}(q_{j+1}) \cup \mathtt{Out}(q_{j+2}) \cup \mathtt{Out}(q_{j+1}) \\
                              &= \mathtt{Out}(q_{j+1}) \cup \mathtt{Out}(q_{j+2}) \\
                              &= \mathtt{Out}(q_j).
        \end{align*}
        Suppose there exists some letter $b\in \mathtt{Out}(q_i)$ with $b\notin \mathtt{q_j}$. Then $x = b$ or $y = b$ as $q_i$ and $q_j$ share the output words of the later states. This is a contradiction to $x = y ? \varepsilon$. So $\mathtt{Out}(q_i) = \mathtt{Out}(q_j)$. But then, as $q_i \rightarrow q_j$, we have a contradiction to $q_i\in M$. Hence, no case from (4.1.1) to (4.1.5) can be true. Thus, (4.1) cannot be true.
    \end{itemize}
    \item[(4.2)] So, finally, suppose $j = i+3$, i.e., there exist two distinct states between $q_i$ and $q_j$. The following possible cases emerge. Recall that both states, $q_j$ and $q_{j-1}$ have an outgoing path to the state $q_\ell$ in the current case (4). For readability purposes, we call the states between $q_i$ and $q_j$ by $p_1$ and $p_2$ with $q_i < p_1 < p_2 < q_j$. Two main cases emerge. Either $p_2 \rightarrow q_j$ or $p_2 \not\rightarrow q_j$. Notice that $p_2 = q_{j-1}$ and therefore $p_2 \rightarrow_* q_\ell$ and $q_j \rightarrow_* q_\ell$ by assumption in these cases.
    \begin{itemize}
        \item[(4.2.1)] Suppose $p_2 \rightarrow q_j$. The following sub-cases emerge.
        \begin{itemize}
            \item[(4.2.1.1)] Suppose $q_i \not\rightarrow_* q_j$. Then, by bandwidth $2$, we must have $q_{i-1} \rightarrow_* q_j$. As $q_i \not\rightarrow_* q_j$, we must have $q_{i-1} \rightarrow p_1$ and $p_1 \rightarrow_* q_j$. But then $p_1 \rightarrow_* q_j$ and $p_2 \rightarrow q_j$ by assumption. As $q_i\in M$, it needs some outgoing transition. But any outgoing transition would go into $p_1$ or $p_2$ due to bandwidth $2$. A contradiction. So, $q_i \rightarrow_* q_j$.
            \item[(4.2.1.2)] Suppose $q_i \not\rightarrow p_2$. Then $q_i \rightarrow p_1$. From (4.2.1.1), we obtain $p_1 \rightarrow_* q_j$. As $q_i,q_j\in M$, we must have $q_1 \overset{\varepsilon}{\rightarrow} p_1$. But then $\mathtt{In}(q_i) \subseteq \mathtt{In}(p_1)$ and $\mathtt{Out}(q_i) = \mathtt{Out}(p_1)$. As $q_i \rightarrow p_1$, this is a contradiction to $q_i\in M$. So, we have $q_i \rightarrow p_2$.
            \item[(4.2.1.3)] Next, suppose $p_1 \not\rightarrow q_j$. As $q_i \rightarrow_* q_j$, $q_i \rightarrow p_2$, $p_2 \rightarrow q_j$, and $q_i,q_j\in M$, we have $q_i \overset{\varepsilon}{\rightarrow_*} p_2$. Suppose, we do not have $p_1 \rightarrow p_2$. Then $p_2$ has only one incoming $\varepsilon$-transition from $q_i$. As $q_j\in M$ and $q_i \rightarrow p_2 \rightarrow q_j$, we must also have $p_2 \overset{\varepsilon}{\rightarrow_*} q_j$. But then $\mathtt{In}(q_i) = \mathtt{In}(q_j)$ and $\mathtt{Out}(q_j) \subseteq \mathtt{Out}(q_i)$. But then, by the definition of $M$, we obtain: If $\mathtt{Out}(j) \subset \mathtt{Out}(q_i)$, then $q_j\notin M$. If on the other hand $\mathtt{Out}(q_j) = \mathtt{Out}(q_i)$, then $i\notin M$. Both are a contradiction. Hence, we must have $p_1 \rightarrow p_2.$. As $p_1 \rightarrow p_2$ is the only outgoing transition from $p_1$, and it reads an $\varepsilon$ as $q_j\in M$ and $2 \rightarrow q_j$, we must have $\mathtt{In}(q_i)\subseteq \mathtt{In}(2)$ and $\mathtt{Out}(q_i) = \mathtt{Out}(p_1) = \mathtt{Out}(p_2)$. This is a contradiction to $q_i\in m$ as $q_i \rightarrow p_2$. Hence, the assumption of this case cannot hold and, thus, we must have $p_1 \rightarrow q_j$.
            \item[(4.2.1.4)] Finally, consider the states $q_i$ and $p_1$.

                (I) Suppose there exists a non-$\varepsilon$-transition $q_i \overset{u}{\rightarrow} p_1$. As $i\in M$, we must have $u\notin V_1$. However, as $p_1 \rightarrow_* q_j$, we can also not have $u\in V_2$. Therefore this transition cannot exist.\\
                (II) Suppose there exists no transition between $q_i$ and $p_1$, i.e., $q_i \not\rightarrow p_1$. Then, again, $p_2$ is the only immediatly reachable state from $q_i$ and as established before, that transition from $q_i$ to $p_2$ can only read $\varepsilon$. This results in $\mathtt{In}(q_i) \subseteq \mathtt{In}(p_2)$ and $\mathtt{Out}(q_i) = \mathtt{Out}(p_2)$. A contradiction to $q_i\in M$.\\
                (III) So, suppose there exists an $\varepsilon$-transition from $q_i$ to $p_1$, i.e., $q_i \overset{\varepsilon}{\rightarrow} p_1$. Then $\mathtt{Out}(q_i) = \mathtt{Out}(p_1) \cup \mathtt{Out}(p_2)$. Note that $\mathtt{Out}(p_1) = \mathtt{Out}(p_2)\cup\mathtt{Out}(q_j)$. We obtain
                \begin{align*}
                    \mathtt{Out}(q_i) &= \mathtt{Out}(p_1) \cup \mathtt{Out}(p_2)\\
                                      &= \mathtt{Out}(p_2)\cup\mathtt{Out}(q_j)\cup \mathtt{Out}(p_2) \\
                                      &= \mathtt{Out}(p_2)\cup\mathtt{Out}(q_j) \\
                                      &= \mathtt{Out}(p_1).
                \end{align*}
                Also, we have $\mathtt{In}(q_i) \subseteq \mathtt{In}(p_1)$ as $q_i \rightarrow p_1$. But then, again, we have a contradiction to $q_i\in M$. Thus, no combination regarding (4.2.1.4) works out. Hence, the whole sub-case (4.2.1) collapses and we must have $p_2 \not\rightarrow q_j$.

        \end{itemize}
        \item[(4.2.2)] So, from now on, assume $p_2 \not\rightarrow q_j$. As $A$ has bandwidth $2$ and $q_j$ needs an incoming transition, this means, that there must exist a transition $p_1 \rightarrow q_j$. We continue with a similar case distinction as in (4.2.1).
        \begin{itemize}
            \item[(4.2.2.1)] Suppose $q_i \not\rightarrow p_1$, i.e., there is no transition from $q_i$ to $p_1$. Then there only exist transitions from $q_i \rightarrow p_2$, due to bandwidth $2$. As $p_2 \rightarrow_* q_ell$ by assumption and as $q_i,q_\ell\in M$, we must have $q_i \overset{\varepsilon}{\rightarrow_*} p_2$. But then, again, $\mathtt{In}(q_i) \subseteq \mathtt{In}(p_2)$ and $\mathtt{Out}(q_i) = \mathtt{Out}(p_2)$. Hence, as $q_i \rightarrow p_2$, this is a contradiction to $q_i\in M$. Thus, we must have $q_i \rightarrow p_1$ and as $p_1 \rightarrow q_j$, we must have $q_i \overset{\varepsilon}{\rightarrow} q_i \overset{\varepsilon}{\rightarrow} q_j$.
            \item[(4.2.2.2)] Suppose $q_i \not\rightarrow p_2$. Then there must exist a transition $p_1 \rightarrow p_2$ as $A$ has no unreachable states and the bandwidth is bound by $2$. As $q_i\in M$, $q_i \rightarrow p_1 \rightarrow p_2 \rightarrow_* q_\ell$, and $q_\ell\in M$, we must have $p_1 \overset{\varepsilon}{\rightarrow} p_2$. We observe $\mathtt{In}(q_i) \subseteq \mathtt{In}(p_1)$ and $\mathtt{Out}(q_i) = \mathtt{Out}(p_1)$, which again is a contradiction to $q_i\in M$. Hence, $q_i \rightarrow p_2$.
            \item[(4.2.2.3)] Finally, suppose $p_1 \rightarrow p_2$. Then, $p_1 \overset{\varepsilon}{\rightarrow} p_2$ as $q_i\in M$ and $q_\ell\in M$ and $q_i \rightarrow p_1 \rightarrow p_2 \rightarrow_* q_\ell$. But then $\mathtt{Out}(q_i) = \mathtt{Out}(p_1) \cup \mathtt{Out}(p_2)$ and $\mathtt{Out}(p_1) = \mathtt{Out}(p_2) \cup \mathtt{Out}(q_j)$. This results in
            \begin{align*}
                \mathtt{Out}(q_i) &= \mathtt{Out}(p_1) \cup \mathtt{Out}(p_2) \\
                                 &= \mathtt{Out}(p_2) \cup \mathtt{Out}(q_j) \cup \mathtt{Out}(p_2) \\
                                 &= \mathtt{Out}(p_2) \cup \mathtt{Out}(q_j) \\
                                 &= \mathtt{Out}(p_1).
            \end{align*}
            As we also have $\mathtt{In}(q_i) = \mathtt{In}(p_1)$, we again have a contradiction to $q_i\in M$. Hence, $p_1 \not\rightarrow p_2$
        \end{itemize}
        But then, as $q_i \overset{\varepsilon}{\rightarrow} p_1 \overset{\varepsilon}{\rightarrow} q_j$ and $p_2 \not\overset{\varepsilon}{\rightarrow} q_j$, we have $\mathtt{In}(q_i) = \mathtt{In}(q_j)$ and $\mathtt{Out}(q_j) \subseteq \mathtt{Out}(q_i)$ as well as $q_i \rightarrow_* q_j$. If $\mathtt{Out}(q_j) \subset \mathtt{Out}(q_i)$, we have a contradiction for $q_j\in M$. Otherwise, if $\mathtt{Out}(q_j) = \mathtt{Out}(q_i)$, we have a contradiction for $q_i\in M$. A contradiction in general in this case. So (4.2.2) cannot hold and we must have $p_2 \rightarrow q_j$. 
    \end{itemize}
    The in (4.2.2) established fact that $p_2 \rightarrow q_j$ is a contradiction to the in (4.2.1) established fact that $p_2 \not\rightarrow q_j$ must be the case. Hence, we obtain a general contradiction if $j = i+3$ in the case of $q_\ell$ being reachable from $q_j$ and $q_{j-1}$. Together with (4.1) and (1), the case of $q_\ell$ being reachable from both states $q_j$ and $q_{j-1}$ leads overall to a contradiction if $j - i < 4$.
\end{itemize}
\end{itemize}
So, for all considered cases, bringing together $(1)$, $(2)$, $(3)$, and $(4)$, we see that if $j-i < 4$, then the state $q_\ell\in M$ with $\ell > j$ cannot exist. This is a a contradiction to the existence of $q_{j_2}$ in the second pair of close by states established before the case distinction. Thus, we must have $|M| \leq c$. With Lemma~\ref{proof:theorem:minization-hard-fixed-k:lemma:M-to-Bicliques} we obtain that $G$ must have a biclique-edge-covering $C$ with at most $|C| = |M| \leq c$ many bicliques. This concludes this direction of the proof. \qed

\end{document}